\pdfoutput=1
\documentclass[a4paper,reqno]{amsart}

\usepackage[T1]{fontenc}
\usepackage{graphicx} 
\usepackage[utf8]{inputenc}

\usepackage[pdftex,linkcolor=black,pdfborder={0 0 0}]{hyperref} 

\usepackage{tikz,pgfplots}
\usetikzlibrary{patterns}
\usetikzlibrary{positioning}
\usetikzlibrary{arrows.meta, decorations.markings, bending}
\usetikzlibrary {arrows.meta,bending}
\usetikzlibrary{decorations.pathmorphing}
\usetikzlibrary{patterns.meta}
\usetikzlibrary{patterns}
\usepackage{tikz-3dplot}

\usepackage{bm}

\usepackage{braket}
\usepackage{bbold}
\usepackage{wrapfig}
\usepackage{float}
\usepackage{soul} 
\usepackage{ulem} 

\usepackage{calrsfs}
\DeclareMathAlphabet{\pazocal}{OMS}{zplm}{m}{n}

\usepackage{mathrsfs} 
\usepackage[backend=biber,style=alphabetic,maxnames=9,maxalphanames=5,isbn=false,backref]{biblatex}
\DeclareFieldFormat[misc]{title}{\mkbibquote{#1}}
\DeclareFieldFormat[book]{title}{\mkbibquote{#1}}
\DeclareFieldFormat{journaltitle}{\mkbibitalic{#1}}
\DeclareFieldFormat{booktitle}{\mkbibitalic{#1}}

\usepackage{graphicx} 
\usepackage{calc} 
\usepackage{enumitem} 

\usepackage[english]{babel}

\usepackage{wrapfig}
\usepackage{amsmath}
\usepackage{amssymb}
\usepackage[dvipsnames]{xcolor}
\usepackage{amsthm}
\usepackage{ytableau}
\usepackage{relsize}
\usepackage{mathtools}

\usepackage{cancel}

\usepackage{caption}
\usepackage{subcaption}

\usepackage{cleveref}
\crefname{lemma}{Lemma}{Lemmas}
\crefname{proposition}{Proposition}{Propositions}
\crefname{definition}{Definition}{Definitions}
\crefname{theorem}{Theorem}{Theorems}
\crefname{corollary}{Corollary}{Corollaries}
\crefname{conjecture}{Conjecture}{Conjectures}
\crefname{claim}{Claim}{Claims}
\crefname{section}{Section}{Sections}
\crefname{appendix}{Appendix}{Appendices}
\crefname{figure}{Fig.}{Figs.}
\crefname{table}{Table}{Tables}

\usepackage[a4paper, left=25mm, right=20mm, top=25mm, bottom=30mm, marginpar=25mm]{geometry} 

\usepackage[protrusion=true,expansion=true]{microtype} 

\usepackage{mathtools}

\hypersetup{ 	
pdfsubject = {},
pdftitle = {},
pdfauthor = {}
}

\usepackage[foot]{amsaddr}

\newtheorem{definition}{Definition}[section]
\newtheorem{problem}{Problem}[section]

\newtheorem{theorem}{Theorem}[section]
\newtheorem{corollary}{Corollary}[section]
\newtheorem{lemma}{Lemma}[section]

\newtheorem{conjecture}{Conjecture}[section]

\newcommand{\eat}[1]{} 

\newcommand{\Chi}{\scalebox{1.3}{$\chi$}}
\renewcommand{\l}{\lambda}

\newcommand{\partition}[1]{\buildrel #1\over\vdash}

\newcommand{\N}{\mathbb{N}}
\newcommand{\Z}{\mathbb{Z}}
\newcommand{\R}{\mathbb{R}}
\newcommand{\Co}{\mathbb{C}}

\renewcommand{\l}{\lambda}

\newcommand{\C}{\pazocal{C}}
\newcommand{\E}{\mathcal{E}}
\newcommand{\id}{\mathbb{1}}
\newcommand{\Tr}{\mathrm{Tr}}

\newcommand{\Del}{\mathrm{Del}}

\newcommand{\im}{\mathrm{i}}

\newcommand{\sgn}{\operatorname{sgn}}
\newcommand{\al}{\alpha}
\newcommand{\Rl}{\pazocal{R}}

\newcommand{\DL}{\mathrm{d}_\mathrm{L}} 
\newcommand{\DP}{\mathrm{d}_\mathrm{P}} 

\newcommand{\AddRef}[1]{{\color{Plum}{[REF]}\normalcolor}} 

\title[PI codes: a numerical study and qudit constructions]{Permutation-invariant codes: a numerical study and qudit constructions}
\author{Liam J. Bond$^{1,2}$}
\author{Ji\v{r}\'{i} Min\'{a}\v{r}$^{1,2}$}
\author{Maris Ozols$^{1,3,4}$}
\author{Arghavan Safavi-Naini$^{1,2}$}
\author{Vladyslav Visnevskyi$^{1,2,5}$}

\address[1]{QuSoft, Amsterdam}
\address[2]{Institute of Physics, University of Amsterdam}
\address[3]{Institute for Logic, Language and Computation, University of Amsterdam}
\address[4]{Korteweg-de Vries Institute for Mathematics, University of Amsterdam}
\address[5]{Centre for the Mathematics of Quantum Theory, University of Copenhagen}

\email{l.j.bond@uva.nl}
\email{j.minar@uva.nl}
\email{marozols@gmail.com}
\email{safavinaini@gmail.com}
\email{vladislav.visnevskyy@gmail.com}

\begin{document}
\begin{abstract}
    We investigate Permutation-Invariant (PI) quantum error-correcting codes encoding a logical qudit of dimension $\DL$ in PI states using physical qudits of dimension $\DP$.
    We extend the Knill--Laflamme (KL) conditions for $d-1$ deletion errors from qubits to qudits and investigate numerically both qubit ($\DL = \DP = 2$) and qudit ($\DL > 2$ or $\DP > 2$) PI codes. We analyze the scaling of the block length $n$ in terms of the code distance $d$, and compare to existing families of PI codes due to Ouyang, Aydin--Alekseyev--Barg (AAB) and Pollatsek--Ruskai (PR).
    Our three main findings are:
    (i) We conjecture that qubit PI codes correcting up to $d-1$ deletion errors have block length $n(d) \geq (3d^2 + 1) / 4$, which implies an upper bound $d \leq \sqrt{12n-3}/3$ on their code distance, and that PR codes can saturate this bound.
    (ii) For qudit PI codes
    encoding a single qudit we numerically observe that increasing $\DP$ results in $n$ monotonically decreasing and approaching the quantum Singleton bound $n(d) \geq 2d-1$.
    (iii) We propose a semi-analytic extension of the qubit AAB construction to qudits that finds explicit solutions by solving a linear program. 
    Our results therefore provide key insights into lower bounds on the block length scaling of both qubit and qudit PI codes, and demonstrate the benefit of increased physical local dimension in the context of PI codes. 
\end{abstract}

\tdplotsetmaincoords{70}{110}

\maketitle
\setcounter{tocdepth}{1}
\tableofcontents

\section{Introduction}\label{ch:Intro}
\subsection{Background}
Quantum computers promise to solve certain problems faster than classical ones. For example, the problem of factoring can be solved in exponentially faster runtime~\cite{Shor_1997}, and the problem of unstructured search in quadratically faster runtime~\cite{grover1996fastquantummechanicalalgorithm}. However, quantum computers are vulnerable to decoherence and quantum noise which limits their scalability, both in terms of the number of qubits and the number of gate operations. Solving this challenge will require the use of quantum error correction (QEC) to protect quantum information from noise-induced errors. 

At the heart of QEC is the quantum error-correcting code (QECC), whereby a logical qubit is encoded across many physical qubits in such a way that errors can be detected and corrected. Experimental realizations of quantum error correction beyond break-even, i.e.\ the regime in which the coherence time of an actively error-corrected logical qubit exceeds that of any qubit without active error correction, has been experimentally demonstrated in quantum hardware platforms including neutral atoms~\cite{Bluvstein_2023}, trapped ions~\cite{egan2021fault, dasu2025breakingmagicdemonstrationhighfidelity} and superconducting circuits~\cite{ofek2016extending, 2024, brock2024quantum}. Thus far, both experimental and theoretical efforts have primarily focused on stabilizer codes, such as surface codes, and bosonic codes, such as GKP codes. However, there exists an entire zoo of QECCs~\cite{albert2023errorcorrectionzoo}, each with different properties, for example in terms of the types of errors that are correctable or in the number of physical qubits that are required. 

Among the many classes of QECCs, Permutation-Invariant (PI) codes feature some key advantages over more well-known codes, such as stabilizer codes. First introduced by Ruskai and Pollatsek~\cite{pollatsek2004permutationally}, PI codes are defined by logical qubit states that are invariant under permutation of the physical qubits. PI codes are able to correct non-Pauli errors such as amplitude-damping errors~\cite{ouyangamplitudedamping}, as well as $t$-qubit deletion errors~\cite{hagiwara2020fourqubitscodequantumdeletion, nakayama2020singlequantumdeletionerrorcorrecting, Ouyang_2021_deletions}, which are $t$-qubit erasure errors whose locations are unknown. This is due to their permutation invariance, which implies that we can simply assume that the first $t$ qubits were erased, and therefore deletion errors on PI codes are equivalent to erasure errors. 
This distinguishes PI codes from stabilizer codes, which cannot correct deletion errors~\cite{Fowler_2012}. Some PI codes have exotic transversal gates that comprise the largest single-qubit gate set implementable transversally \cite{kubischta2023family}, with more efficiently distillable magic states. PI codes are also well-suited for realization on current quantum hardware platforms. They can be efficiently encoded using quantum circuits of one- and two-qubit gates~\cite{Bartschi2019} or with global variational quantum circuits~\cite{BondPRR2025,GutmanPRL2024,ZouPRA2003}, and appear naturally as the ground states of Heisenberg ferromagnetic Hamiltonians~\cite{ouyang2014permutation, aydin2024family}. PI codes have been recently shown to be equivalent to both Fock state codes and spin codes~\cite{aydin2025quantumerrorcorrectionsu2}, and therefore they are immediately applicable to quantum hardware platforms comprised of many qubits, large spins or bosonic modes. Recent works also present efficient algorithms for error-correction on PI codes~\cite{ouyang2026theoryquantumerrorcorrection}. Beyond quantum error correction, PI codes have also found recent application in quantum storage~\cite{ouyang2021storage}, quantum metrology~\cite{ouyang2212finite}, and quantum communication~\cite{bhalerao2025improvingquantumcommunicationrates}.

Previous works have introduced qubit PI code constructions that are either fully~\cite{aydin2024class, ouyang2014permutation, ouyang2017permutation, aydin2024family} or partially~\cite{pollatsek2004permutationally, aydin2025quantumerrorcorrectionsu2} analytic. However, the optimality of these constructions remains an open problem. In the first part of this manuscript, we numerically investigate the smallest block length $n$ required to construct error-correcting qubit PI codes, as a function of the code distance $d$, physical dimension $\DP$ and logical dimension $\DL$. The properties of all PI codes investigated in this paper are summarized in \cref{tab:codes_summary}.

A natural generalization of qubit PI codes is to extend their definition to qudits. A qudit PI code is then defined by logical qudits encoded in permutationally-invariant states of physical qudits. Due to their larger state space that is available for storing and processing quantum information \cite{wang2020qudits}, qudit systems may offer advantages over qubit systems, for example, in terms of scaling for realizing universal quantum computation~\cite{kiktenko2020scalable}, or a reduction in the elementary gate count required to construct gates such as the Toffoli gate~\cite{ralph2007efficient}.  Qudit systems can also be natively realized in experimental platforms when more than two physical levels are available, including in photons~\cite{milburn2009photons}, superconducting systems~\cite{chiorescu2003coherent}, trapped ions~\cite{blatt2008entangled} and neutral atoms~\cite{saffman2010quantum}. 

Qudit systems may also allow for more efficient logical state encoding in terms of the number of local physical systems that are required to construct a QECC. Strategies to reduce the overhead of QECCs are especially important for current experimental platforms, which are limited in their scalability. Although qudit PI code constructions are known due to Ouyang~\cite{ouyang2014permutation}, the number of physical qubits (block length) that they require to correct $t$ deletion errors is independent of the physical qudit dimension. That is, Ouyang qudit PI codes have no benefits to increasing the physical state space. In the second part of this manuscript we therefore ask, are there qudit PI codes whose block length, $n$, decreases with increasing physical qudit dimension, $\DP$? We numerically investigate the smallest $n$ required to construct error-correcting qudit PI codes, and find that for logical dimensions $\DL=2$, $\DL=3$ and $\DL=4$, $n$ decreases with increasing $\DP$. Our numerical results therefore provide a pathway to realizing qubit and qudit PI codes with a smaller physical overhead. 

Below, we provide an overview and summary of the results in this manuscript. Note that while nearing the completion of this manuscript, we became aware of a related work independently deriving Knill--Laflamme conditions for qudit PI codes, and constructing qudit PI codes using convex geometry~\cite{aydin2025quantumerrorcorrectionsu2}.

\subsection{Overview and summary of the results}

After introducing notation in \cref{sec:notation}, in \cref{ch:Numerical_study} we investigate qubit PI codes, $\DL = \DP = 2$. For complex codeword coefficients, we determine numerically the minimal block length, $n_{\rm min}$, as a function of the code distance $d$. Based on this numerical evidence we conjecture that: 
\begin{equation}\label{eq:overview_minblock}
    n_{\rm min}(d) = (3d^2+1)/4, 
\end{equation}
which scales as $n=\frac{3}{4}d^2 + O(d)$. In comparison, to the best of our knowledge the best existing analytic qubit PI code constructions in terms of the scaling $n(d)$ achieve $n=d^2 + O(d)$~\cite{aydin2024family,ouyang2014permutation}. We further conjecture that qubit PI codes with real coefficients obey the same minimal block length scaling as \cref{eq:overview_minblock}. For such real, minimal block length qubit PI codes, we observe that they: (i) feature symmetries in the codeword coefficients which allow the coefficients of the second codeword to be determined completely from the first, and (ii) have a solution space which, for any fixed $d$, is a $\lfloor (d+1)/2 \rfloor$-dimensional manifold. 

Next, in \cref{ch:KL_qudit_gen} we generalize the Knill--Laflamme (KL) conditions~\cite{knill1997theory}, also known as the quantum error correction conditions, for qubit PI codes~\cite{aydin2024family} to qudit PI codes. In \cref{sec:converting_qubit_to_qudit} we then use these conditions to show analytically that any qudit PI code on qubits can be extended to a qudit PI code on qudits without any loss in the code distance. In \cref{sec:num_study_qudit_PI}, we set $d=3$ and numerically investigate the scaling of the minimal block length $n_{\rm min}(\DP)$ with physical dimension $\DP$ for various logical dimensions $\DL$. We observe that $n_{\rm min}$ decreases with $\DP$, approaching the quantum Singleton bound~\cite{Knill_2000}. In comparison, to the best of our knowledge the best existing analytic qudit PI code constructions in terms of the scaling $n(\DP)$ are independent of $\DP$~\cite{ouyang2014permutation}. Outside of asymptotic limits~\cite{aydin2025quantumerrorcorrectionsu2}, our numerical results are the first demonstration for PI codes that increasing the physical dimension allows smaller block lengths. 

Finally, we provide a semi-analytic extension of the AAB qubit PI code family~\cite{aydin2024family} to qudits using discrete simplices, which we call simplicial PI codes (\cref{sec:Simplicial_codes}). For $\DP = 3$ we observe numerically that our construction scales as $n(d) = (d^3 + 23d^2 - 9d - 7)/8$, which is worse than existing analytic constructions due to Ouyang that have $n(d) = d^2(\DL-1)$. However, this result for $\DP=3$ does not imply the scaling cannot improve for simplicial codes of higher values of $\DP$. Even though the final scaling of our construction is suboptimal for $\DP=3$ compared to existing state-of-the-art qudit PI codes, insights from our extension may lead to explicit analytic constructions featuring subquadratic block length scaling with code distance.

Importantly, all our results for both qubit and qudit PI codes immediately apply to Fock state codes and spin codes due to the equivalence established in~\cite{aydin2025quantumerrorcorrectionsu2}.

\begin{table}[H]
    \centering
    \begin{tabular}{|c|c|c|c|c|}
         \hline
         & & & & \\
         \textbf{PI Code} & \textbf{$\DL$} & \textbf{$\DP$} & \textbf{Block length scaling $n(d)$} & \\
         & & & & \\
         \hline
         & & & & \\
         Minimal (\cref{def:minimal_PI_codes}) & $\DL = 2$ & $\DP = 2$ & $n(d) = (3d^2 + 1)/4$ & N \\
         & & & & \\
         \hline
         & & & & \\
         AAB (\cref{sec:Aydin_codes}, \cite{aydin2024family}) & $\DL = 2$ & $\DP = 2$ & $n(d) = d^2 -d + 1$ & A \\
         & & & & \\
         \hline
         & & & & \\
         PR (\cite{pollatsek2004permutationally}) & $\DL = 2$ & $\DP = 2$ & $n(d) = (3d^2 + 1)/4$ & N \\
         & & & & \\
         \hline
         & & & & \\
         Ouyang (Example 6.3 in \cite{ouyang2017permutation}) & $\DL \in \N$ &
         $\DP \in \N$ & $n(d) = d^2 (\DL-1)$ & A \\
         & & & & \\
         \hline
         & & & & \\
         Simplicial (\cref{sec:Simplicial_codes}) & $\DL \in \N$ & $\DP \geq \DL$ &$\DP = 3$: $n(d) = (d^3 + 23d^2 - 9d - 7)/8$ & N \\
         & & & & \\
         \hline
    \end{tabular}
    \caption{Summary of the properties of PI codes reviewed in this manuscript. Here $\DL$ is the logical local dimension, $\DP$ is the physical local dimension, and $d$ is the code distance. Each of these PI codes encodes a single qubit/qudit. All qubit PI codes in this table have mirror and phase-flip symmetries, see \cref{eqn:mirror_phase_sym}. The PR code scaling was determined in this work numerically. In the last column, ``N'' and ``A'' denote whether the scaling was obtained numerically or analytically.}
    \label{tab:codes_summary}
\end{table}

\section{Preliminaries}
\label{sec:notation}

\subsection{Notation and basic definitions}
Throughout the paper, we use the term \textbf{local dimension} to describe the arity of a single quantum register. For instance, a qubit has local dimension $2$, a qutrit has local dimension $3$, and so on.
Qudit error-correcting codes encode \textbf{logical qudits} of local dimension $\bm{\DL}$ in \textbf{physical qudits} of local dimension $\bm{\DP}$. Additionally, if not specified otherwise, whenever we denote a quantum state by $\ket{x}$ we assume it is a basis state of the computational basis.
We use the following notation and definitions throughout:
\begin{itemize}
    \item $\lambda \partition{q} n := (\lambda_0,\lambda_1,\dots,\lambda_{q-1})$ s.t.\ all $\lambda_i \geq 0$ and $\sum_{i=0}^{q-1} \lambda_i = n$ -- a \textit{composition} of $n$ into $q$ parts (composition is an ordered partition), and $\sum_{\lambda \partition{q} n}$ denotes the sum ranging over all such compositions;
    \item $l(\lambda)$ -- length of the composition $\lambda$, e.g.\ for $\lambda \partition{q} n$ we have $l(\lambda)=q$;
    \item $[q] := \{0,1,\dots,q-1\}$ -- the set of all integers from $0$ to $q-1$.
\end{itemize}

\begin{definition}
Let $x \in [q]^n$ be a string of length $n$ over
alphabet $\{0,1,\dots,q-1\}$.
The \textbf{weight} of $x$ is the vector
\begin{equation}
    w(x) := (N_0(x),N_1(x),\dots,N_{q-1}(x))
\end{equation}
where $N_i(x)$ is the number of entries $i \in [q]$ in $x$, e.g.\ $N_2(021230) = 2$.
For a given composition $\lambda \partition{q} n$, we denote the set of all weight-$\lambda$ strings of length $n$ by
\begin{equation}
    X^n_\lambda := \{x \in [q]^n \mid w(x) = \lambda \}.
\end{equation}
Note that
$|X^n_\lambda| = \binom{n}{\lambda} := \frac{n!}{\lambda_0! \lambda_1! \dots \lambda_{l(\lambda)-1}!}$ is the multinomial coefficient.
\end{definition}

\begin{definition}\label{def:Dicke_state}
    A Dicke state $\ket{D^n_{\lambda}}$ is a normalized, permutation-invariant state given by the uniform superposition over all $n$-qudit states of weight $\lambda$:
    \begin{equation}
       \ket{D_\lambda^n} = \frac{1}{\sqrt{\binom{n}{\lambda}}} \sum_{x=(x_1,\dots,x_n) \in X_\lambda^n} \ket{x_1}\otimes \ket{x_2}\otimes \cdots \otimes\ket{x_n}.
    \end{equation}
\end{definition}

 We refer to the $\l$ index of the Dicke state also as the \textbf{weight} of the Dicke state. Dicke states of different weight within the same Hilbert space are orthogonal, $\braket{D_\lambda^n|D^n_{\lambda'}} = \delta_{\lambda, \lambda'}$. The Dicke states obtained from all compositions $\lambda \partition{\DP} n$ form a complete basis for the \textbf{symmetric subspace} of $n$ qudits of local dimension $\DP$.

\begin{definition}\label{def:qudit_PI_codes}
    A Permutation-Invariant (PI) qudit code encoding $k$ logical qudits of local dimension $\DL$ into $n$ physical qudits of local dimension $\DP$ is a subspace of $(\Co^{\DP})^{\otimes n}$ spanned by codewords of the form
    \begin{equation}
        \ket{c_i} = \sum_{\lambda \partition{\DP} n} \alpha_{i,\lambda} \ket{D_\lambda^n},
    \end{equation}
    where $i \in \{0,1,\dots,\DL-1\}^k$ and $\alpha_{i,\lambda} \in \Co$.
\end{definition}

\subsection{QEC conditions and errors}
In quantum error correction the computational space is encoded into an open system in which errors occur due to an unwanted interaction with the environment. As with any quantum evolution, an error acting on a quantum state can be described by a quantum channel, also called error channel. The general noise model does not make any assumptions about the nature of errors: the span of Kraus operators of an error channel is called the error set, and an error can be any element of this set.

Within this general noise model, for a quantum code $\C$ to be error correcting on an error set $\E$ it must satisfy certain necessary and sufficient quantum error-correction conditions, called the Knill--Laflamme (KL) conditions \cite{gottesmansurviving, knill1997theory}:
\begin{theorem}\label{thrm:Knill--Laflamme}
    $(\C, \E)$ is a quantum error-correcting code iff $\forall \ket{\psi},\ket{\phi} \in \C, \forall E_a, E_b \in \E$:
        \begin{equation}
            \bra{\psi} E_a^{\dagger} E_b \ket{\phi} = C_{ab} \braket{\psi|\phi},
        \end{equation}
    where $C_{ab} \in \Co$ is a constant that only depends on $E_a$ and $E_b$ but not on $\ket{\psi}, \ket{\phi}$.
\end{theorem}
\begin{proof}
    See proof of Theorem 2.7 in \cite{gottesmansurviving}.
\end{proof}

\noindent Informally, the intuition behind this condition is that (i) orthogonal codewords should stay orthogonal even after errors are introduced so as to be detectable and distinguishable, and (ii) the errors themselves, being general linear maps, should act as unitaries when restricted to the codespace in order to be reversible.

For a quantum code $\C$ encoding $k$ logical qudits into $n$ physical qudits, we refer to $n$ as the \textbf{block length} of this code. In this manuscript we set $k = 1$ always. 

To formalize the notion of distance of a quantum error-correcting code, we say that a multi-qudit Pauli operator $E$ on the physical space has \textbf{weight $t$} if it acts non-trivially on at most $t$ physical qudits.
Errors of weight $t$ are then spanned by Pauli operators acting on at most $t$ qudits. Due to linearity of quantum error correction, a code correcting $t$ Pauli errors can thus correct any weight-$t$ error.

\begin{definition}
    The distance $d \in \N$ of a code $\C$ is the smallest weight of an operator $F$ such that    \begin{equation}\label{eqn:code_distance}
        \bra{\psi}F\ket{\phi} = c(F) \braket{\psi|\phi},
    \end{equation}
    is violated by some codewords $\ket{\psi},\ket{\phi} \in \C$, where $c(F)$ is independent of $\ket{\psi},\ket{\phi}$.
\end{definition}

Since a code of distance $d$ satisfies \cref{eqn:code_distance} for all operators of weight at most $d-1$, the condition of \cref{thrm:Knill--Laflamme} holds for all errors of weight $\lfloor (d-1)/2 \rfloor$ or less. 

Besides Pauli noise, a particularly important type of errors in QEC are \textbf{erasure errors}. A single-qudit erasure error maps a single-qudit state $\ket{\psi} \in \pazocal{H} = \Co^{\DP}$ to some state $\ket{s}$ that is orthogonal to the space spanned by the qudit's basis states, i.e.\ $\ket{s} \bot \, \pazocal{H}$. Then, by making a measurement to detect the presence of $\ket{s}$ we can determine whether an erasure error has occurred. A multi-qudit erasure error is a tensor product of single-qudit erasure errors on different qudits, with all local states $\ket{s_i}$ on different qudits mutually orthogonal. Then, by locally measuring on each qudit whether an erasure has occurred we are able to determine which qudits were erased, gaining classical information about the locations of the erasure errors. 

\begin{theorem}\label{thrm:QECC_corr_eras}
    A QECC with distance $d$ can correct $d-1$ erasure errors.
\end{theorem}
\begin{proof}
    See proof of Theorem 2.10 in \cite{gottesmansurviving}.
\end{proof}
Because any distance $d$ code can correct at least $\lfloor \frac{d-1}{2} \rfloor$ errors,
\cref{thrm:QECC_corr_eras} immediately implies an important equivalence:
\begin{corollary}\label{cor:2tErasure-tPauli_equiv}
    A QECC can correct arbitrary $t$-qudit errors if and only if it can correct 2$t$ erasure errors.
\end{corollary}

Although erasure errors provide classical information about the location of the error, using PI codes to correct erasures does not require the use of this information because the code construction is permutation-invariant \cite{Ouyang_2021_deletions}. For PI codes, an erasure error is therefore equivalent to a \textbf{deletion error}. Informally a deletion error, also called qudit loss, is an erasure error that does not reveal classical information about the location of the erasure. For example an erasure on the third qubit maps $1001 \mapsto 10s1$, while a deletion on the second or third qubit maps $1001 \mapsto 101$. Rigorously:
\begin{definition}
    A weight-$t$ deletion map $\Tr_I$ on an index set $I = \{i_1,i_2, \dots, i_t\}$ with $i_1<i_2< \dots < i_t$ is defined as a channel that traces out the systems at positions specified by $I$:
    \begin{equation}
        \Tr_I (\rho) = \Tr_{i_1} \circ \Tr_{i_2} \circ \dots \circ \Tr_{i_t} (\rho).
    \end{equation}
\end{definition}
We can now define the deletion channel:
\begin{definition}
    A $t$-deletion channel $\Del_t^n$ on $n$ systems is defined as a convex linear combination of all weight-$t$ deletion maps $\Tr_I$ with index sets of size $t$, i.e.\ with $|I|=t$:
    \begin{equation}
        \Del_t^n(\rho) = \sum_{I:|I|=t} p(I) \Tr_I(\rho),
    \end{equation}
    where $p(I)$ is a probability distribution.
\end{definition}
\noindent We say that $t$ deletion errors happened when a $t$-deletion channel acts on a state, and if a QECC can recover the state after a $t$-deletion channel, it can correct up to $t$ deletion errors.

\begin{theorem}[\cite{Ouyang_2021_deletions}, \cite{aydin2024family}]\label{thrm:PI_deletion_equivalence}
    A PI QECC can correct arbitrary $t$-qudit errors if and only if it can correct up to 2$t$ deletion errors.
\end{theorem}
\begin{proof}
Irrespective of the choice of the index set $I$ of length $t$, all weight-$t$ deletion maps $\Tr_I$ act identically on Dicke states due to their permutation invariance. Thus, when restricting the deletion map to the symmetric subspace, all reduced states $\Tr_I(\rho)$ in the convex sum are equal. Therefore, $\Del_t^n(\rho)$ is equivalent to acting with $\Tr_I(\rho)$ on some fixed index set $I$ with $|I|=t$, so we can treat the general $t$-deletion channel as if we are deleting at known locations. This implies that when acting on PI codes deletion errors are equivalent to erasure errors.

As established earlier in \cref{cor:2tErasure-tPauli_equiv}, a QECC can correct arbitrary $t$-qudit errors iff it can correct up to 2$t$ erasures. Because deletion and erasure errors are equivalent for PI codes, this concludes the proof.
\end{proof}


\section{Numerical study of qubit PI codes}\label{ch:Numerical_study}
In this section, we numerically investigate qubit PI codes. Based on evidence obtained by numerically minimizing a cost function defined in terms of the KL conditions from \cite{aydin2024family}, we conjecture that the minimal scaling of the block length with the weight $t$ of correctable errors is $n_{\rm min}(t) = 3t^2 + 3t + 1$.\footnote{Note that we switched from using the code distance $d$ to an equivalent description through the error weight $t$, with $d = 2t+1$.} We also conjecture that every real solution whose block length is minimal obeys a ``mirror'' and ``phase-flip'' symmetry, and that the family of real minimal qubit PI codes is a continuous manifold of dimension $t+1$. 

\subsection{Minimal block length scaling of qubit PI codes}\label{sec:qubit_PI_numeric_properties}

For the case of qubit PI codes, i.e.\ $\DP=\DL=2$, the KL conditions are~\cite{aydin2024family}:
\begin{subequations}\label{eqn:KL_PI_conds_qubits}\begin{align}
    &\text{C1: } \sum_{j=0}^n \bar{\alpha}_j \beta_j = 0, \; \; \sum_{j=0}^n |\alpha_j|^2 = \sum_{j=0}^n |\beta_j|^2 =1,\\
    &\text{C2: } \sum_{j=0}^n \bar{\alpha}_j \beta_{j-a+b} \frac{\binom{n-2t}{j-a}}{\sqrt{\binom{n}{j} \binom{n}{j-a+b}}} = 0, \; \; \forall a,b \in \N, \; 0 \leq a,b \leq 2t,\\
    &\text{C3: } \sum_{j=0}^n (\bar{\alpha}_j \alpha_{j-a+b} - \bar{\beta}_j \beta_{j-a+b}) \frac{\binom{n-2t}{j-a}}{\sqrt{\binom{n}{j} \binom{n}{j-a+b}}} = 0, \; \; \forall a,b \in \N, \; 0 \leq a,b \leq 2t,
\end{align}\end{subequations}
where we define $\frac{\binom{n-2t}{j-a}}{\sqrt{\binom{n}{j} \binom{n}{j-a+b}}}$ to be zero if any of the binomial coefficients contain factorials of negative numbers. For notational simplicity, here we denote the coefficients of the 0th codeword (see \cref{def:qudit_PI_codes}) as $\alpha_{j}:=\alpha_{0,(n-j,j)}$, and the coefficients of the 1st codeword as $\beta_{j}:=\alpha_{1,(n-j,j)}$. We also use $a:=\mu_1$ for $\mu=(\mu_1,\mu_2) = (2t-a,a)$, and similarly $b:=\nu_1$ for $\nu=(\nu_1,\nu_2) = (2t-b,b)$. 

\begin{definition}\label{def:minimal_PI_codes}
    For PI codes encoding a logical qubit into $n$ physical qubits and correcting all Pauli errors of weight $t$, the block length $n(t)$ is \textbf{minimal} if there exist no PI codes of block length smaller than $n$ correcting $t$ errors. We denote the minimal block length for a given number of errors $t$ as $n_{\rm min}(t)$. We call the subfamily of PI codes that for any $t \in \N$ have block length $n_{\rm min}(t)$ \textbf{minimal PI codes}.
\end{definition}

To determine $n_{\rm min}(t)$, we numerically search for solutions to the KL conditions of \cref{eqn:KL_PI_conds_qubits} by minimizing the following cost function:
\begin{align}\begin{split}
    f_{\rm cost}(\alpha,\beta) ={} & (\alpha^\dagger \alpha-1)^2 + (\beta^\dagger \beta-1)^2 + (\alpha^\dagger \beta)^2  \\ 
    &+ \sum_{a,b} \left( \sum_j b_{a,b,j} \bar{\alpha}_j \beta_{j-a+b}  \right)^2 + \sum_{a,b} \left(\sum_j b_{a,b,j} [\bar{\alpha}_j \alpha_{j-a+b} - \bar{\beta}_j \beta_{j-a+b}] \right)^2.
\end{split}\label{eq:f_cost}\end{align}
In words, the cost function is a sum of squares of the $2(2t+1)^2 +3$ identities in \cref{eqn:KL_PI_conds_qubits}. The minimization of $f_{\rm cost}$ is performed numerically in \texttt{Julia}~\cite{Julia-2017} using \texttt{Optim.jl}~\cite{mogensen2018optim}, see Ref.~\cite{GitHub_Visnevskyi} for scripts and data files. In general, the optimization landscape is non-convex, and therefore gradient descent methods can become trapped in local minima. To mitigate this issue we repeat the minimization up to $1000$ times, starting from different random initial guesses for the codewords. We consider a solution found when the cost function has (i) converged, which we impose by demanding that the gradient is $< 10^{-20}$, and (ii) sufficiently small, which we define as $f_{\rm cost} < 10^{-18}$. The optimization results for codewords with unconstrained complex coefficients for $t = 1$, $t=2$ and $t=3$ are shown in \cref{fig:NumericsOptimization}(a). For each of these error weights, solutions are first found for $n(t) = 7$, $n(t) = 19$ and $n(t) = 37$, respectively. Note that our minimal block length results agree with Ref.~\cite{aydin2024family}, who numerically find no PR codes correcting $t=2$ Pauli errors for $n(t=2)<19$. We therefore conclude that these are the minimal block lengths. 

\begin{figure}
    \centering
    \includegraphics{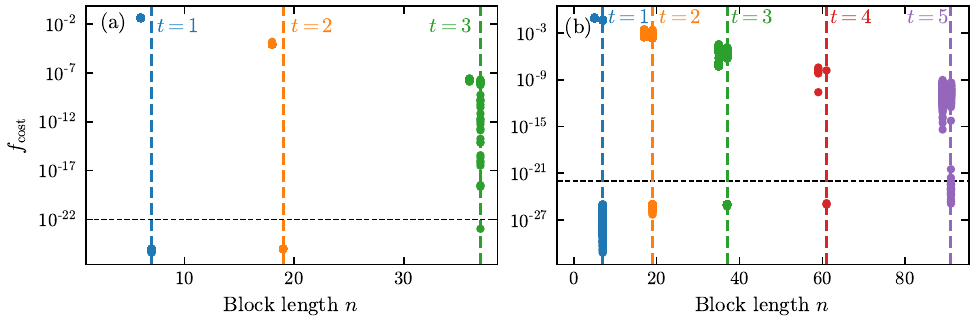}
    \caption{
    Numerical optimization of the cost function of \cref{eq:f_cost} with (a) codewords with complex coefficients and (b) codewords restricted to the PR family~\cite{pollatsek2004permutationally}. For each $n$ and $t$ the optimization is run up to $1000$ times (solid points), with each repetition using a maximum of $1.2\times10^6$ optimization steps. Vertical dashed lines show, for each $t$, the smallest $n$ for which a solution exists, i.e.\ when $f_{\rm cost} < 10^{-18}$ (black dashed line). This is the minimal block length, $n_{\rm min}(t)$. We find $n_{\rm min}(t=1)=7$, $n_{\rm min}(t=2)=19$ and $n_{\rm min}(t=3)=37$ for both codeword families. Due to the reduced computational resources required, we are able to obtain $n_{\rm min}(t=4)=61$ and $n_{\rm min}(t=5)=91$ for the PR code family.
    }
    \label{fig:NumericsOptimization}
\end{figure}

To obtain solutions for $t > 3$, we reduce the number of optimization variables from $2(n+1)$ to $(n+1)/2$ by restricting to the case of codewords in the Pollatsek--Ruskai (PR) family, which are defined by the conditions $\ket{c_1} = (\bigotimes_j X_j) \ket{c_0}$, $(\bigotimes_j Z_j) \ket{c_0} = \ket{c_0}$ and $(\bigotimes_j Z_j) \ket{c_1} = -\ket{c_1}$. The numerical optimization results are shown in \cref{fig:NumericsOptimization}(b). For $t=1,2,3$ we find solutions at $n_{\rm min}(t) = 7,19,37$, which is identical to the minimal block lengths obtained for the case of unrestricted complex coefficients. This implies that PR codes are minimal qubit PI codes, and that complex codeword coefficients offer no benefits in terms of block length scaling compared to real codeword coefficients. The reduced number of optimization parameters also allows us to find solutions for $t = 4$ and $t = 5$, at $n_{\rm min}(t) = 61$ and $n_{\rm min}(t) = 91$, respectively. 

\begin{figure}
    \centering
    \includegraphics{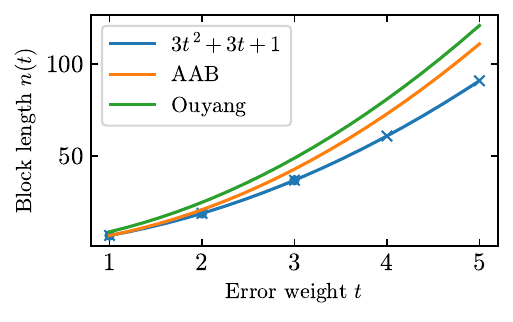}
    \caption{
    Block length $n(t)$ versus error weight $t$ for minimal PI codes with complex coefficients (blue dots) and for minimal PR PI codes~\cite{pollatsek2004permutationally} (blue crosses) determined numerically at $n_{\rm min} = 7,19,37,61,91$ for $t=1,2,3,4,5$, c.f.\ \cref{fig:NumericsOptimization}. The quadratic fit (solid blue line) to $n_{\rm min}(t) = 3t^2 + 3t+1$ demonstrates perfect agreement, which forms the basis of~\cref{con:min_PI_qubit_scaling}. We also plot the block length scaling of the AAB~\cite{aydin2024family} (orange line) and Ouyang~\cite{ouyang2014permutation} (green line) code families.
    }
    \label{fig:QubitQubitPICodes}
\end{figure}

In \cref{fig:QubitQubitPICodes} we summarize our numerical results by showing the minimal block length scaling $n_{\rm min}(t)$ with the error weight $t$. Fitting a quadratic function yields perfect agreement with $n_{\rm min}(t) = 3t^2 + 3t+1$. Existing qubit PI codes with analytically defined coefficients have block length scalings of $n(t) = 4t^2 + 4t + 1$~\cite{ouyang2014permutation} and $n(t) = 4t^2 + 2t + 1$~\cite{aydin2024family}. These analytic constructions therefore upper bound the minimal block length scaling. Conversely, the quantum Singleton bound~\cite{Knill_2000} implies a lower bound on the block length:
\begin{equation}
    \pazocal{H}_L \leq \DP^{n-2d+2}, \; t:=\lfloor\frac{d-1}{2}\rfloor
    \quad\Rightarrow\quad
    \pazocal{H}_L \leq \DP^{n-4t}
    \quad\Rightarrow\quad
    n \geq 4t + 1, \label{eq:Singleton}
\end{equation}
where the last equality follows from the fact that encoding at least one logical qudit requires the Hilbert space dimension $\pazocal{H}_L$ to be greater than one. Based on our numerical results of \cref{fig:QubitQubitPICodes} and further informed by these upper and lower bounds, we make the following conjecture:
\begin{conjecture}\label{con:min_PI_qubit_scaling}
    Minimal qubit PI codes have block length scaling $n_{\rm min}(t) = 3t^2+3t+1$. That is, all qubit PI quantum error-correcting codes have an upper bound on their code distance, $d \leq \sqrt{12n-3}/3$, that is saturated in minimal PI codes.
\end{conjecture}

Note that \cref{con:min_PI_qubit_scaling} implies that there are no good qLDPC~\cite{Breuckmann_2021, eczoo_good_qldpc} qubit PI codes, because even for qubit PI codes that encode a single logical qubit the best distance scaling is $d = O(\sqrt{n})$, and therefore linear asymptotic distance in qubit PI codes cannot be achieved if \cref{con:min_PI_qubit_scaling} holds.

\subsection{Codeword symmetries}\label{sec:codewordsymmetries}

For all real minimal qubit PI codes found numerically, we observe the following symmetry in the codeword coefficients: 
\begin{align}\label{eqn:mirror_phase_sym}
    \beta_i = (-&1)^{i} \alpha_{n_{\rm min}-i},
\end{align}
where $n_{\rm min}$ is the minimal block length for an error weight $t$, and $i \in [n+1]$. One example solution is shown in \cref{fig:PhaseFlipSymmetry}. We call the $|\beta_i| = |\alpha_{n-i}|$ symmetry \textbf{mirror symmetry}, and the $\sgn(\beta_i) = (-1)^{i} \sgn(\alpha_{n-i})$ \textbf{phase-flip symmetry}. We also identify that the $\beta$ (or $\alpha$) coefficients can differ by a global minus sign, i.e.\ given a solution $\beta$ (or $\alpha$), it follows immediately from the KL conditions of \cref{eqn:KL_PI_conds_qubits} that $-\beta$ (or $-\alpha$) is also a solution.  We note that the AAB~\cite{aydin2024family}, the Pollatsek--Ruskai~\cite{pollatsek2004permutationally} and the Ouyang~\cite{ouyang2014permutation} code families all obey the same symmetries as a consequence of the structure that they enforce by construction. Based on this numerical evidence, we make the following conjecture:
\begin{conjecture}\label{con:mirror_phase_sym}
    All minimal qubit PI codes with real coefficients exhibit both the mirror and the phase-flip symmetries \eqref{eqn:mirror_phase_sym} in their codewords.
\end{conjecture}
Note that this implies that minimal PI codes with real coefficients are defined by the $n+1$ coefficients of the first codeword only. 

\begin{figure}
    \centering
    \includegraphics[width=0.5\textwidth]{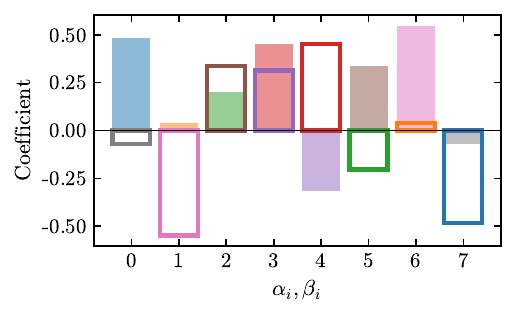}
    \caption{
    Codeword coefficients $\alpha_i$ for $\ket{c_0}$ (solid bars) and $\beta_i$ for $\ket{c_1}$ (hollow bars) for one example real minimal qubit PI code for block length $n = 7$ and error weight $t = 1$. The absolute value of the codeword coefficients have mirror symmetry with respect to reversing the composition labels, as emphasized by the bar coloring. We also observe that every second coefficient has the opposite sign compared to its mirror pair, starting at $\alpha_0$ (for $\ket{c_0}$) and $\beta_1$ (for $\ket{c_1}$). These mirror and phase-flip symmetries are defined in \cref{eqn:mirror_phase_sym}.
    }
    \label{fig:PhaseFlipSymmetry}
\end{figure}

\subsection{Codeword continuity}\label{sec:CodewordContinuity}

\begin{figure}
    \centering
    \includegraphics{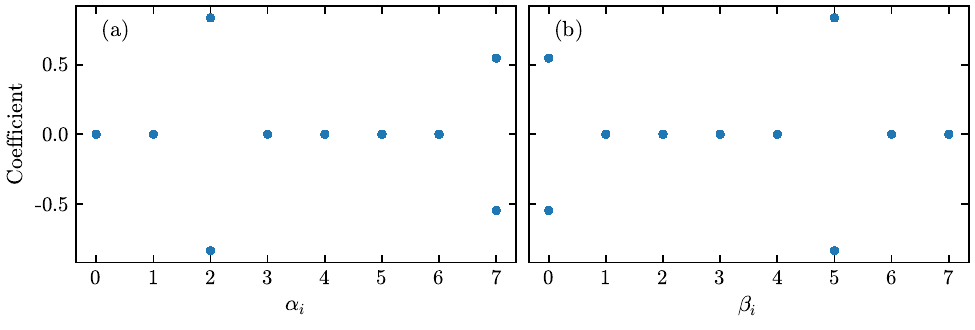}
    \caption{
    Codeword coefficients $\alpha_i$ for $\ket{c_0}$ (a) and $\beta_i$ $\ket{c_1}$ (b) for real minimal qubit PI codes with block length $n = 7$ and error weight $t = 1$. We fix $M = 2$ coefficients in the first codeword, i.e.\ $\alpha_0 = \alpha_1 = 0$. We numerically minimize the cost function of \cref{eq:f_cost} using the remaining $14$ codeword coefficients. Repeating the minimization $1000$ times from random initial choices for the non-fixed codeword coefficients, we converge to the same eight solutions shown here. A valid solution exists for both positive and negative signs of each of the $\alpha_2$, $\alpha_7$, $\beta_0$ and $\beta_5$, giving a total of eight solutions. 
    }
    \label{fig:ParamFixing}
\end{figure}

Next, we investigate the number of degrees of freedom of real minimal qubit PI codes. To do so, we fix $M$ coefficients in the $\ket{c_0}$ codeword, and numerically optimize the remaining $2(n+1)-M$ coefficients. For the case of $(n,t) = (7,1)$ with $M = 2$ coefficients fixed, after numerically minimizing the cost function in \cref{eq:f_cost} starting from $1000$ random initial guesses for the remaining (non-fixed) codeword coefficients, we observe that the optimization converges to finitely many solutions. In \cref{fig:ParamFixing} we plot the codewords found when fixing the first two parameters to zero, i.e.\ setting $\alpha_0 = \alpha_1 = 0$. Notice that the $\alpha_2$ and $\alpha_7$ coefficients are either positive or negative, and that due to the symmetry of~\cref{con:mirror_phase_sym} the $\beta_0$ and $\beta_5$ coefficients are also either positive or negative. We therefore find only eight solutions for this fixing of $\alpha_0 = \alpha_1 = 0$.

\begin{figure}
    \centering
    \includegraphics[width=0.95\linewidth]{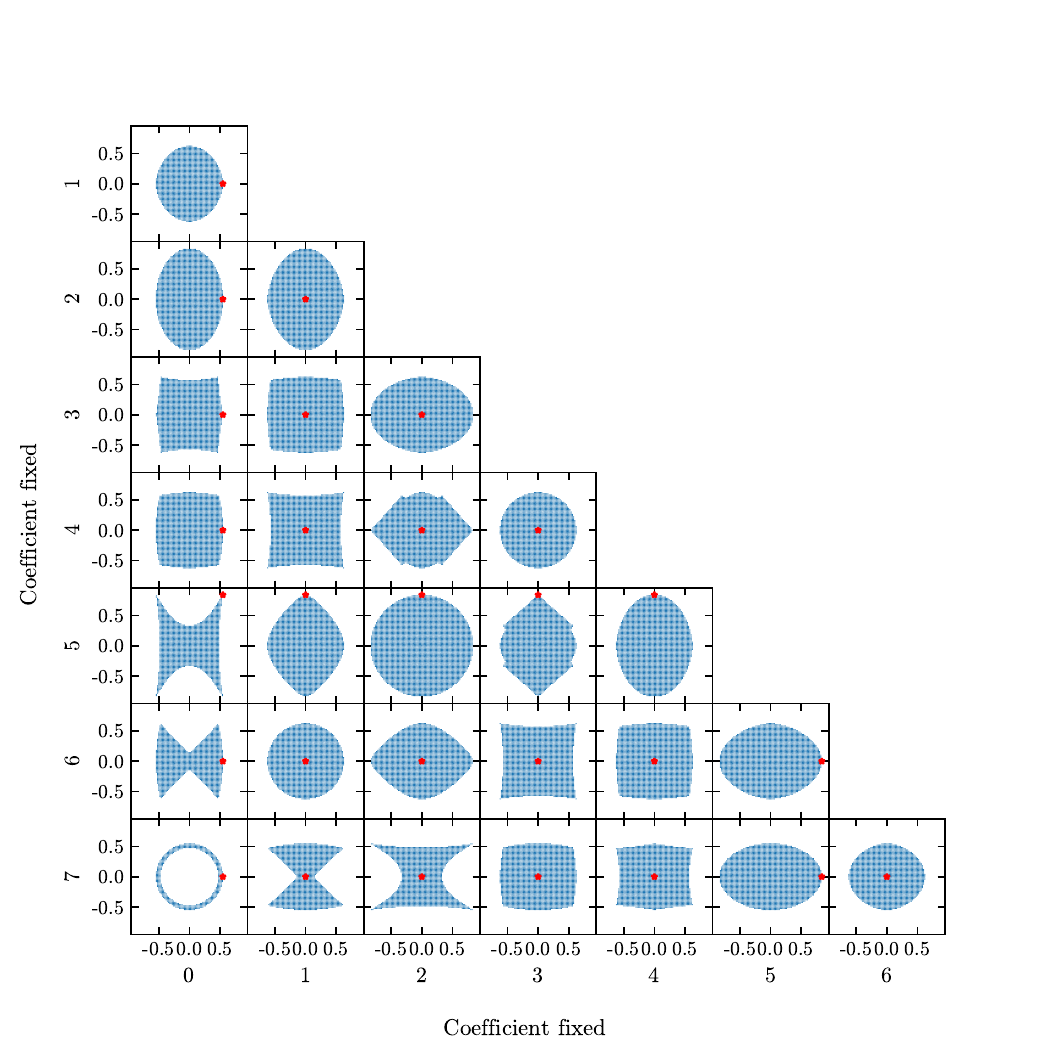}
    \caption{
    Table of all solution space projections on coordinate planes defined by the $26$ possible pairs of $\ket{c_0}$ codeword coefficient fixings, $\alpha_i,\alpha_j$ for $0 \leq i < j \leq 7$, found numerically for real minimal qubit PI codes with $(n,t)=(7,1)$. We do not enforce the phase-flip nor mirror symmetry. The grid search is performed only in the top left quadrant with $0 \leq \alpha_i,\alpha_j \leq 0.85$ and with the constraint $\alpha_i^2 + \alpha_j^2 \leq 1$. The remaining three quadrants are constructed using the symmetries observed in \cref{sec:codewordsymmetries}. The solution corresponding to the AAB code family of Ref.~\cite{aydin2024family} is indicated by the red star. 
    }
    \label{fig:Manifold_plots_table}
\end{figure}

We further investigate real minimal PI codes for $(n,t)=(7,1)$ with $M=2$ fixed coefficients by iterating through the $28$ possible choices for which two coefficients in $\ket{c_0}$ to fix. Given the $i$th and $j$th coefficients are fixed, we numerically minimize the cost function of \cref{eq:f_cost} while iterating $\alpha_i$ and $\alpha_j$ over a grid subject to the constraints $\alpha_i,\alpha_j \geq 0$ and $\alpha_i^2 + \alpha_j^2 \leq 1$ in a step size of $0.02$. The results are shown in \cref{fig:Manifold_plots_table}, with values of $\alpha_i,\alpha_j$ for which a solution exist represented by a point. We only search in the upper right quadrant of the plane spanned by $\alpha_i$ and $\alpha_j$ because we can immediately obtain a solution in the lower left quadrant using the fact that the sign of the fixed coefficients can be reversed due to the global minus sign symmetry of the KL conditions of \cref{eqn:KL_PI_conds_qubits}, which means that $\alpha \rightarrow -\alpha$ is also a solution. We obtain the remaining two quadrants due to the phase-flip and mirror symmetries of \cref{eqn:mirror_phase_sym}. Finally we note that the symmetry of flipping along the diagonal of the triangle plot in \cref{fig:Manifold_plots_table} is due to the qubit Dicke state symmetry under a global $\pi/2$ rotation about $X$, i.e.\ $\alpha_i \rightarrow \alpha_{n-i}$ and $\beta_i \rightarrow \alpha_{n-i}$. 

In \cref{fig:Manifold_plots_table}, we observe that there are densely packed (in the sense of the resolution of the numerical scan) regions in which we find solutions. Due to the finite resolution of our scan, this is an approximate reconstruction of the full space of real minimal qubit PI codes for $(n,t)=(7,1)$. The high density of the scan, and the fact that solutions are found at every grid point within certain boundary regions, suggests that the solution set may be continuous, and therefore we refer to it as a \textbf{solution manifold}. Note that this would imply that the family of minimal qubit PI codes is infinite. 

We also investigate the number of degrees of freedom for error weights $t=2$ and $t=3$ with $M=3$ and $M=4$ coefficients fixed to $\alpha_0,\alpha_1,\alpha_2 = 0$ and $\alpha_0,\alpha_1,\alpha_2,\alpha_3 = 0$, respectively. To reduce the computational cost, for these error weights we impose the mirror and phase-flip symmetries of \cref{eqn:mirror_phase_sym}. After repeating the numerical minimization $20\times10^4$ and $5.6 \times 10^4$ times for $(M,t) = (2,3)$ and $(M,t) = (3,4)$, we find eight and $16$ solutions, respectively. We plot the solutions in Appendix.~\ref{paramfixing_t2t3}. From these numerical examples with small $t$, we observe that $M=t+1$ coefficients should be fixed to undergo a transition from infinitely- to finitely-many solutions. That is, minimal qubit PI codes have $t+1$ continuous degrees of freedom. We therefore make the following conjecture:

\begin{conjecture}\label{con:algebraic_variety}
    For any fixed error weight $t$ and block length $n=n_{\rm min}(t)$, the family of \textbf{real} minimal qubit PI quantum error-correcting codes with parameters $(n,t)$ is a manifold of dimension $t+1$.
\end{conjecture}

Finally, we note that in some cases it is possible to find analytical solutions from the numerical ones. Focusing on cases where a certain number of codeword coefficients are close to zero, it is possible to reconstruct the non-zero ones exactly. For example, for $(n,t)=(7,1)$ we choose $(\alpha_0,\alpha_7) = (\frac{\sqrt{15}}{10},\frac{\sqrt{15}}{10})$ to be the top right point of the hollow disc in \cref{fig:Manifold_plots_table}. For this point we find that $\al_1=\al_3=\al_4=\al_6=0$. We finally reconstruct the non-zero values from the numerics as $\al_2=-\al_5=\frac{\sqrt{35}}{10}$.
While this approach allows us to obtain analytical solutions for the codewords from numerical results, it becomes challenging with an increasing number of (non-zero) variables and therefore we leave it as a matter for future study.

\section{Knill--Laflamme conditions for qudit PI codes}\label{ch:KL_qudit_gen}
In this section, we use the equivalence between erasure and deletion errors that is present in PI codes, extending the deletion channel construction for PI quantum codes given in \cite{aydin2024family} from qubits to qudits. Using this, we then derive the necessary and sufficient conditions for arbitrary qudit PI quantum codes to be error-correcting.

\subsection{Kraus set for a qudit quantum deletion channel}

Aydin, Alekseyev, and Barg~\cite{aydin2024family} use a Kraus decomposition for a $t$-deletion channel to derive the KL conditions for qubit PI codes. We extend this construction to qudit PI codes.

To derive the Kraus decomposition, in spirit of \cite{shibayama2021equivalence}, we want to define an operator with the following properties: for an index set $I=\{i_1,i_2,\dots,i_t\} \subseteq \{1,\dots,n\}$ of size $t$  with $i_1<i_2< \dots < i_t$, and a $t$-qudit computational basis state $\bra{x}$ with $x = x_1x_2\dots x_t \in \{0,1,\dots,q-1\}^t$, we have the operator $A_{I,x}^n$ acting on $n$ qudits that has the following tensor product structure
\begin{align}\label{eqn:A_prop_1}
    A_{I,x}^n = A_{I,1} \otimes A_{I,2}\otimes \dots \otimes A_{I,n}, 
    \end{align}
where each $A_{I,j}$ is defined as,
    \begin{equation}\label{eqn:A_prop_2}
        A_{I,j} = 
        \begin{cases}
            \bra{x_k}, & j=i_k \in I; \\
            \id_q, &j \notin I.
        \end{cases}
    \end{equation}
For example, $A_{I,i_3} = \bra{x_3}$, and $A_{I,k} = \id_q$ if $k \notin I$. Then, to give a somewhat informal definition for $A_{I,x}^n$ satisfying \cref{eqn:A_prop_1,eqn:A_prop_2}, we have:
\begin{definition}
    Let $I \subseteq \{1,\dots,n\}$ and $x \in \{0,1,\dots,q-1\}^{|I|}$. We define $A_{I,x}^n : \Co^{q^n} \to \Co^{q^{n-|I|}}$ as the operator that projects qudits located at positions $I$ to the standard basis state $\bra{x}$:
     \begin{equation}
        A_{I,x}^n := \bra{x}_I \otimes \id_{\{1,\dotsc,n\} \setminus I},
    \end{equation}
    where terms in the tensor product need to be rearranged appropriately. In this particular instance, by $\id_{\{1,\dotsc,n\} \setminus I}$ we mean an identity of size $q*(n-|I|)$ that is acting on appropriately rearranged positions.
\end{definition}

Similarly to the qubit case  \cite{shibayama2021equivalence}, for qudits $A_{I,x}^n$ also has a decomposition:
\begin{lemma}\label{lem:Kraus_op_decomp}
    The operator $A_{I,x}^n$ can be decomposed as:
    \begin{equation}
        A_{I,x}^n = A_{i_1,x_1}^{n-t+1} \circ A_{i_2,x_2}^{n-t+2} \circ \dots \circ A_{i_t,x_t}^{n}.
    \end{equation}
\end{lemma}
\begin{proof}
    Straightforward computation, see proof in Appendix~\ref{apx:proofs}.
\end{proof}
\noindent Now we observe a Kraus decomposition:
\begin{lemma}
    We have the following Kraus decomposition for the $t$-deletion channel:
    \begin{equation}
        \Del_t^n(\rho) = \sum_{x \in \{0,\dots,q-1\}^t; \; I:|I|=t} p(I)  A_{I,x}^n \rho A_{I,x}^{n \dagger},
    \end{equation}
    with $\sqrt{p(I)} A_{I,x}^n$ -- the Kraus operators.
\end{lemma}
\begin{proof}
    A simple calculation gives:
    \begin{equation}
        \Tr_I (\rho) = \sum_{x \in \{0,\dots,q-1\}^t} A_{i_1,x_1}^{n-t+1} \circ \dots \circ A_{i_t,x_t}^{n} \rho A_{i_t,x_t}^{n \dagger} \circ \dots \circ A_{i_1,\ket{x_1}}^{n-t+1 \dagger} = \sum_{x \in \{0,\dots,q-1\}^t} A_{I,x}^n \rho A_{I,x}^{n \dagger},
    \end{equation}
    which, when introducing a convex combination over $\Tr_I$, implies the statement of the lemma, concluding the proof. 
\end{proof}
Although $\sqrt{p(I)} A_{I,x}^n$ are the actual Kraus operators, we will often refer to the normalized $A_{I,x}^n$ as the Kraus operators.

\begin{definition}
    A $k$-type deletion operator applied to the $i$-th qudit in an $n$-qudit system is the operator:
    \begin{equation}
        \Chi_{i,k} = A_{i,k}^n,
    \end{equation}
\end{definition}
\noindent so that the action of $\Chi_{i,k}$ is:
\begin{equation}
    \Chi_{i,k} \ket{x} = \braket{k|x_i}\ket{\cancel{x}_i},
\end{equation}
where $\cancel{x}_i$ is the string $x$ with the $i$-th entry deleted, i.e.\ for a size $n$ string $x = (x_1,\dots x_{i-1}, x_i, x_{i+1},\dots, x_n)$ we have $\cancel{x}_i:= (x_1,\dots, x_{i-1}, x_{i+1},\dots, x_n)$. Note that the index $n$ is unnecessary in $\Chi_{i,k}$ since the index set is always of size $1$.

\begin{lemma}\label{lem:k-del_on_Dicke_state}
    The action of a $k$-deletion $\Chi_{i,k}$ on a Dicke state $\ket{D^n_{\l}}$ is:
    \begin{equation}
        \Chi_{i,k} \ket{D_{\l}^n} =  \sqrt{\frac{\binom{n-1}{\l-e_k}}{\binom{n}{\l}}} \ket{D_{\l-e_k}^{n-1}},
    \end{equation}
    where $\l - e_k =(\l_0, \dots, \l_k -1, \dots, \l_{q-1})$ for $e_k$ the $k$-th standard basis vector. 
\end{lemma}
\begin{proof}
     \begin{align}
        \Chi_{i,k} \ket{D_{\l}^n} &= \frac{1}{\sqrt{\binom{n}{\lambda}}} \sum_{x=(x_1,\dots,x_n) \in X_\lambda^n} \braket{k|x_i}  \ket{x_1}\otimes \cdots \otimes \ket{x_{i-1}}\otimes \ket{x_{i+1}} \otimes \cdots \otimes\ket{x_n} \\
        &= \frac{1}{\sqrt{\binom{n}{\lambda}}} \sum_{\substack{x=(x_1,\dots,x_{n-1}) \in X_{\l-e_k}^{n-1}}} \ket{x_1}\otimes \ket{x_2} \otimes \cdots \otimes\ket{x_{n-1}} \\ 
        &= {\frac{\sqrt{\binom{n-1}{\l-e_k}}}{\sqrt{\binom{n}{\l}}}} \ket{D_{\l-e_k}^{n-1}}
    \end{align}
    where, by convention, we say that the Dicke state is zero if the composition denoting this Dicke state contains any negative entries.
\end{proof}
\noindent
Note that when acting with any $k$-type deletion on a Dicke state, the index $i$ in $\Chi_{i,k}$ is irrelevant. Since $i$ denotes a particular qudit which is being deleted and Dicke states are permutation-invariant, the outcome of a $k$-deletion is the same for all indices $i$. \\
Since we will be mostly discussing errors happening on PI and Dicke states, it is convenient to remove the indexation in $i$ when discussing the $k$-type deletions, and denote them as $\Chi_k$ instead.
\begin{lemma}
    For a Dicke state $\ket{D_{\l}^n}$, $\forall a \in \{1,\dots,n\}$:
    \begin{equation}
        \left(\Chi_k\right)^a \ket{D_{\l}^n} := \underbrace{\Chi_k \circ \cdots \circ \Chi_k}_a \ket{D_{\l}^n} = \sqrt{\frac{\binom{n-a}{(\l_0, \dots, \l_k - a, \dots, \l_{q-1})}}{\binom{n}{\l}}} \ket{D_{(\l_0, \dots, \l_k - a, \dots, \l_{q-1})}^{n-a}}
    \end{equation}
\end{lemma}
\begin{proof}
    Follows from \cref{lem:k-del_on_Dicke_state} by induction.
\end{proof}

\begin{lemma}
    Deletion operators commute when acting on Dicke states.
\end{lemma}
\begin{proof}
    Consider a $k$-type deletion $\Chi_k$ and a $l$-type deletion $\Chi_l$, where without loss of generality we have $k<l$ then:
    \begin{equation*}
        \Chi_k \circ \Chi_l \ket{D_{\l}^n} = \sqrt{\frac{\binom{n-2}{\l-e_l-e_k}}{\binom{n}{\l}}} D_{\l-e_l-e_k}^{n-2} =  \sqrt{\frac{\binom{n-2}{\l-e_k-e_l}}{\binom{n}{\l}}} D_{\l-e_k-e_l}^{n-2} = \Chi_l \circ \Chi_k \ket{D_{\l}^n}.\qedhere
    \end{equation*}
\end{proof}
From the lemmas above, it is clear that a general Kraus operator $A_{I,\bra{x}}^n$ decomposes into a product of $k$-type deletions (for different $k$) when acting on Dicke states. For a qudit $t$-deletion channel acting on states with local dimension $q$, this decomposition has the following form:
\begin{equation}
    (\Chi_{q-1})^{\mu_{q-1}} \circ (\Chi_{q-2})^{\mu_{q-2}} \circ \cdots \circ (\Chi_{1})^{\mu_{1}} \circ (\Chi_{0})^{\mu_{0}} \; \mathlarger{\mathlarger{\vert}} \; \mu \partition{q} t,
\end{equation}
where from now on we will label:
\begin{equation}
    E_{\mu}^q = (\Chi_{q-1})^{\mu_{q-1}} \circ \cdots \circ (\Chi_{0})^{\mu_{0}},
\end{equation}
and further on we call $E_{\mu}^q$ \textit{deletion errors}. Therefore, when acting on Dicke states we can write the Kraus set for a qudit $t$-deletion channel acting on states with local dimension $q$ as:
\begin{equation}\label{eqn:del_err_Kraus_set_Dicke}
    \mathcal{E}_t^q = \{E_{\mu}^q\ \vert \; \mu \partition{q} t \}.
\end{equation}
\begin{lemma}\label{lem:del_err_Dicke}
     For a Dicke state $\ket{D_{\l}^n}$ of local dimension $q$, and for $\l \partition{q} n, \mu \partition{q} t, \l - \mu \geq 0$:
    \begin{equation}
        E_{\mu}^q \ket{D^n_{\l}} = \sqrt{\frac{\binom{n-t}{\l - \mu}}{\binom{n}{\l}}} \ket{D_{\l - \mu}^{n-t}},
    \end{equation}
    where $\l - \mu$ denotes the entry-wise difference of $\l$ and $\mu$, and $\l - \mu \geq 0$ means that each entry in the resulting vector is non-negative. The action of deletion errors on Dicke states gives zero if $\l - \mu \geq 0$ doesn't hold.
\end{lemma}
\begin{proof}
    Straightforward, see Appendix~\ref{apx:proofs}.
\end{proof}

\subsection{Necessary and sufficient error-correction conditions}
Using the previous lemmas we can now derive KL conditions for PI codes to correct arbitrary weight-$t$ errors. As a result of \cref{thrm:PI_deletion_equivalence}, we do this using the Kraus operators for the $2t$ deletion channel, i.e.\ we use deletion errors from the Kraus set in  \cref{eqn:del_err_Kraus_set_Dicke}. We use the codewords defined as in \cref{def:qudit_PI_codes}.
\begin{theorem}\label{thrm:PI_KL_conds_qudits}
    The KL conditions for a qudit PI code with parameters $((n,\DL^1,d))_{\DP}$ correcting $t = \lfloor\frac{d-1}{2}\rfloor$ errors are:
    \begin{subequations}
    \label{eqn:KL_qudits}
    \begin{align}
        &\text{C1: } \sum_{\l \partition{\DP} n} \bar{\alpha}_{i, \l} \alpha_{j, \l} = \delta_{i,j}, \\
        &\text{C2: } \sum_{\l \partition{\DP} n} \bar{\alpha}_{i, \l} \alpha_{j, \l-\mu+\nu} \frac{\binom{n-2t}{\l-\mu}}{\sqrt{\binom{n}{\l} \binom{n}{\l-\mu+\nu}}} = 0, \; \;  i\neq j, \\
        &\text{C3: } \sum_{\l \partition{\DP} n} (\bar{\alpha}_{i, \l} \alpha_{i, \l-\mu+\nu} - \bar{\alpha}_{j, \l} \alpha_{j, \l-\mu+\nu}) \frac{\binom{n-2t}{\l-\mu}}{\sqrt{\binom{n}{\l} \binom{n}{\l-\mu+\nu}}} = 0,
    \end{align}
    \end{subequations}
    where there is a condition for each $i,j \in [\DL]$ and each composition $\mu,\nu \partition{\DP} 2t$. The block length must satisfy $n \geq 2t$. We define $\frac{\binom{n-2t}{\l-\mu}}{\sqrt{\binom{n}{\l} \binom{n}{\l-\mu+\nu}}}$ to be zero if any of the multinomial coefficients contain factorials of negative numbers (so that it is zero when $\l - \mu$ contains negative entries).
\end{theorem}
\begin{proof}
    \hfill \break
     C1. This is just the orthonormality condition:
     \begin{equation}
         \braket{c_i | c_j} = \sum_{\l,\l' \partition{\DP} n} \bar{\alpha}_{i, \l} \alpha_{j, \l'} \underbrace{\braket{D_{\l}^n|D_{\l'}^n}}_{\delta_{\l\l'}} = \sum_{\l \partition{\DP} n} \bar{\alpha}_{i, \l} \alpha_{j, \l} = \delta_{i,j}, \; \; \forall i,j \in [\DL].
     \end{equation}
    \noindent C2. Derived from the Knill--Laflamme condition $\bra{c_i} E_{\mu}^{\DP \! \dagger}E_{\nu}^{\DP} \ket{c_j} = 0$ for $i \neq j$, where we have $\mu,\nu \partition{\DP} 2t$ since we are working with the $2t$-deletion channel:
    \begin{align}
        0 &= \bra{c_i} E_{\mu}^{\DP \! \dagger}E_{\nu}^{\DP} \ket{c_j} \\ 
        \implies 0 &= \sum_{\l,\l' \partition{\DP} n} \bar{\alpha}_{i, \l} \alpha_{j, \l'} \braket{D_{\l}^n|E_{\mu}^{\DP \! \dagger}E_{\nu}^{\DP}|D_{\l'}^n} \\ 
        &= \sum_{\l,\l' \partition{\DP} n} \bar{\alpha}_{i, \l} \alpha_{j, \l'}\sqrt{\frac{\binom{n-2t}{\l - \mu}\binom{n-2t}{\l' - \nu}}{\binom{n}{\l}\binom{n}{\l'}}} \underbrace{\braket{D_{\l - \mu}^{n-2t}|D_{\l' - \nu}^{n-2t}}}_{\delta_{\l-\mu+\nu,\l'}}  \\
        &=\sum_{\l \partition{\DP} n} \bar{\alpha}_{i, \l} \alpha_{j, \l-\mu+\nu} \frac{\binom{n-2t}{\l-\mu}}{\sqrt{\binom{n}{\l} \binom{n}{\l-\mu+\nu}}} , \forall i\neq j | i,j \in [\DL], \; \forall \mu,\nu \partition{\DP} 2t, n \geq 2t.
    \end{align}
    Note that only entries with $\l - \mu \geq 0$ remain non-zero because Dicke states which do not satisfy this requirement get deleted (i.e.\ turn to zero) by the deletion errors.
    
    \noindent C3. Derived from the Knill--Laflamme condition $\bra{c_i} E_{\mu}^{\DP \! \dagger}E_{\nu}^{\DP} \ket{c_i} = \bra{c_j} E_{\mu}^{\DP \! \dagger}E_{\nu}^{\DP} \ket{c_j}$, where again $\mu,\nu \partition{\DP} 2t$:
    \begin{align}
    \bra{c_i} E_{\mu}^{\DP \! \dagger}E_{\nu}^{\DP} \ket{c_i} - \bra{c_j} E_{\mu}^{\DP \! \dagger}E_{\nu}^{\DP} \ket{c_j} &= 0 \\ 
    \implies  \sum_{\l \partition{\DP} n} (\bar{\alpha}_{i, \l} \alpha_{i, \l-\mu+\nu} - \bar{\alpha}_{j, \l} \alpha_{j, \l-\mu+\nu}) \frac{\binom{n-2t}{\l-\mu}}{\sqrt{\binom{n}{\l} \binom{n}{\l-\mu+\nu}}} &= 0, 
        &\forall i,j \in [\DL], \; \forall \mu,\nu \partition{\DP} 2t, n \geq 2t. 
    \end{align}
\end{proof}

\section{Study of qudit PI Codes}

\subsection{Extending qudit PI codes on qubits to qudit PI codes on qudits}\label{sec:converting_qubit_to_qudit}
Intuitively, given a PI code encoding a logical qudit with physical qubits, we can treat all errors whose local dimension is higher than two (for example a shift error $X^2 \ket{0} =\ket{2}$) as deletion errors, since taking the state outside of the physical Hilbert space can be interpreted as deleting the qubit. As such, a qudit PI error-correcting code can also correct errors of higher local dimension. Equivalently, a qudit PI code on physical qubits can be extended to a qudit PI code on physical qudits without any loss in the code distance. We make this statement rigorous as follows. 

\begin{theorem}
    Given a $((n,\DL,d))_2$ qudit PI quantum error-correcting code on qubits that can correct $d-1$ deletion errors of local dimension $2$, it can be extended (padded) to a $((n,\DL,d))_{\DP}$ qudit PI quantum error-correcting code of the same code distance $d$ for arbitrary $\DP \in \N$ correcting a larger error set of $d-1$ deletion errors of local dimension $\DP$.
\end{theorem}
Note that our proof below is inspired by a construction of qudit PI codes due to \cite{ouyang2017permutation}, which provided an explicit instance of how qubit codes can be extended to qudits.
\begin{proof}
    For a given block length $n$, logical local dimension $\DL$, physical local dimension $\DP = 2$, and distance $d$, assume that there exists an $((n,\DL,d))_2$ PI code with coefficients $\gamma_{i,\lambda} \in \Co$ for $\lambda \partition{2} n$ and $i \in [\DL]$ that satisfies the Knill--Laflamme equations in \cref{thrm:PI_KL_conds_qudits}. For any $\DP \in \N$ we construct an $((n,\DL,d))_{\DP}$ PI code with codeword coefficients $\alpha_{i,\lambda}$ from the $((n,\DL,d))_2$ PI code using the following assignment:
    \begin{equation}
        \alpha_{i,\lambda} = 
        \begin{cases}
            \gamma_{i,(\lambda_0,\lambda_1)}, \; \lambda = (\lambda_0,\lambda_1,\underbrace{0,0,\dots,0}_{\DP-2}) \\
            0, \; \text{otherwise}.
        \end{cases}
        \label{eq:beta_padding}
    \end{equation}
    We show that under this assignment, the KL conditions for the 
    $((n,\DL,d))_{\DP}$ PI code reduce to the KL conditions for the $((n,\DL,d))_2$ PI code. For the first KL condition, C1, we have: 
    \begin{align}
     \text{C1: } \sum_{\l \partition{\DP} n} \bar{\alpha}_{i, \l} \alpha_{j, \l} = \sum_{\l \partition{2} n} \bar{\gamma}_{i, \l} \gamma_{j, \l} + 0 = \delta_{i,j}, \; \; \forall i,j \in [\DL],
    \end{align} 
where the summation over $\lambda$ reduces to a summation over compositions of $n$ into two parts only due to \cref{eq:beta_padding}. 

For the second condition C2, which must hold for all values of $i \neq j \in [\DL]$, all compositions $\mu, \nu \partition{\DP} 2t$ and for all block lengths $n \geq 2t$ we consider three cases. Firstly, for all compositions $\mu$ whose $\mu_k \neq 0$ for at least one value of $k \in \{3,\dots,\DP-1\}$, and for all compositions $\nu$, 
\begin{align}
     \bar{\alpha}_{i, \l} \alpha_{j, \l-\mu+\nu} \frac{\binom{n-2t}{\l-\mu}}{\sqrt{\binom{n}{\l} \binom{n}{\l-\mu+\nu}}} = 0, \; \forall \lambda \partition{\DP} n  \label{eq:C2Case1},
\end{align}
This is because in the summation over $\lambda$ in C2, the only non-zero terms are the $\lambda$ compositions of $n$ into two parts due to \cref{eq:beta_padding}. For these surviving $\lambda$, for the $\mu$ considered in this case we have that, 
\begin{align}
    \binom{n-2t}{\lambda-\mu} = 0,
\end{align}
because $\lambda$ does not contain any non-zero entries beyond the first two, while $\mu$ contains at least one non-zero entry beyond the first two, and therefore in the tuple $\lambda-\mu$ there is at least one negative entry. 

The second case is for all remaining $\mu$ compositions, i.e.\ those with $\mu_k = 0 \; \forall k \in\{3,\dots\DP-1\}$, and for $\nu$ compositions whose $\nu_k \neq 0$ for at least one value of $k \in \{3,\dots,\DP-1\}$. In this case, $\lambda-\mu+\nu$  will have non-zero entries beyond the first two, and therefore $\alpha_{j, \l-\mu+\nu} = 0$. Thus for the $\mu$ and $\nu$ considered in this case, we have the same equalities as \cref{eq:C2Case1}. 

The third case is the only remaining combination of $\mu$ and $\nu$ compositions, i.e.\ those with $\mu_k, \nu_k = 0 \; \forall k \in\{3,\dots\DP-1\}$. In this case, 
\begin{align}
    \sum_{\l \partition{\DP} n} \bar{\alpha}_{i, \l} \alpha_{j, \l-\mu+\nu} \frac{\binom{n-2t}{\l-\mu}}{\sqrt{\binom{n}{\l} \binom{n}{\l-\mu+\nu}}} = \sum_{\l \partition{2} n} \bar{\gamma}_{i, \l} \gamma_{j, \l-\mu'+\nu'} \frac{\binom{n-2t}{\l-\mu'}}{\sqrt{\binom{n}{\l} \binom{n}{\l-\mu'+\nu'}}} = 0, 
    \end{align}
where the summation over $\lambda$ again reduces due to \cref{eq:beta_padding}. Because these $\lambda, \mu$ and $\nu$ compositions are only those with all elements zero except for the first two entries, so we introduced notation for compositions of length two $\mu',\nu' \partition{2} 2t$, i.e.\ $\mu' = (\mu_0', \mu_1')$ and $\nu' = (\nu_0',\nu_1')$. We also substituted $\alpha$ for $\gamma$ following \cref{eq:beta_padding}. The C2 condition for the $\alpha_{i,\lambda}$ coefficients of the qudit PI code is therefore equivalent to the C2 condition for the $\gamma_{i,\lambda}$ coefficients of the qubit PI code. 

Finally, the third KL condition (C3) for the $((n,\DL,d))_{\DP}$ PI code analogously simplifies when considering the same three cases as in the proof for the second condition (C2) of the $((n,\DL,d))_2$ PI code.

This concludes the proof.
\end{proof}

Our proof shows that qudit PI codes on qubits can always be padded to qudit PI codes on qudits without any loss in the code distance. However, this construction does not benefit from the increased number of degrees of freedom that are available in qudit systems. In the next subsection, we numerically find qudit PI codes whose minimal block length scaling $n_{\rm min}$ reduces with increasing local physical dimension, $\DP$. 

\subsection{Numerical study of qudit PI codes}\label{sec:num_study_qudit_PI}

We investigate the scaling of the minimal block length $n_{\rm min}$ of qudit PI codes with the physical dimension $\DP$. Restricting to the case of errors with weight $t = 1$, we follow the same numerical procedure as described in \cref{ch:Numerical_study}. Specifically, we use the cost function of \cref{eq:f_cost}, but with the KL conditions for qudit PI codes of \cref{eqn:KL_qudits}. The results are shown in \cref{fig:QuditScaling}. 

\begin{figure}
    \centering
    \includegraphics{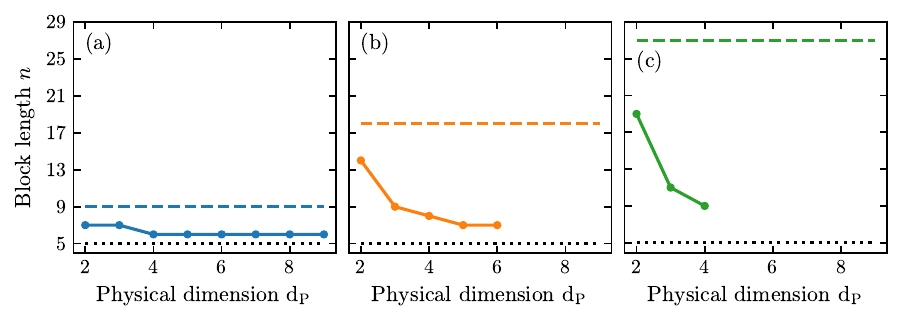}
    \caption{
    Minimal block length $n_{\rm min}$ (colored dots) versus physical qudit dimension $\DP$ for logical dimension (a) $\DP = 2$, (b) $\DP = 3$ and (c) $\DP = 4$, all with error weight $t = 1$. The block length decreases with $\DP$. For the largest $\DP$ accessible with our numerics, we find $n_{\rm min}(\DL=2,\DP=9)=6$, $n_{\rm min}(\DL=3,\DP=6)=7$ and $n_{\rm min}(\DL=4,\DP=4)=9$. We also show the lower bound obtained from the quantum Singleton bound in \cref{eq:Singleton} of $n = 5$ (black dotted line) and the block length for Ouyang qudit PI codes~\cite{ouyang2014permutation} of $n = 9(\DL-1)$ (dashed colored lines), both of which are independent of $\DP$. For $\DL = 4$, Ouyang qudit PI codes requires $n=27$, which is three time larger than the $n = 9$ qudits required for our $\DP = 4$ qudit PI code. 
    }
    \label{fig:QuditScaling}
\end{figure}

In \cref{fig:QuditScaling}(a) we fix the logical dimension $\DL = 2$ and plot the minimal block length $n_{\rm min}$ as a function of the physical dimension $\DP$. We observe that $n_{\rm min}$ decreases with $\DP$, reaching $n_{\rm min}(\DL=2) = 6$ for $4\leq\DP\leq9$, with $\DP=9$ the largest system that we were able to solve with our numerics. We find a minimal block length that is always smaller than the block length of analytic Ouyang qudit PI codes, $n(\DL=2) = 9$ (blue dashed line). Also note that the $n_{\rm min}$ that we find decreases with $\DP$, approaching the quantum Singleton bound c.f. \cref{eq:Singleton} (black dotted line). In contrast, Ouyang PI codes have a block length that is independent $\DP$. In \cref{fig:QuditScaling}(b) we set the logical dimension $\DL = 3$, and in \cref{fig:QuditScaling}(c) we set $\DL = 4$. We again observe decreasing $n_{\rm min}$ with $\DP$. The reduction in $n$ for our qudit PI codes compared to Ouyang qudit PI codes is largest for $\DL = 4$, where we obtain a threefold reduction in block length scaling, from $n = 27$ to $n = 9$. 

From our numerical results, we conclude that the block length can be significantly reduced by increasing the physical qudit dimension $\DP$. We are not aware of any existing analytic PI code families whose block length decreases with $\DP$. 

\subsection{Extending the qudit codes beyond \cref{eq:beta_padding}:}
\label{sec:Simplicial_codes}

To go beyond the constraint of \cref{eq:beta_padding}, we extend the AAB construction, which we review in Appendix~\ref{sec:Aydin_codes}, from qubits to qudits. We define the codewords:
\begin{equation}\label{eqn:Simplicial_codes}
    \ket{c_i} = \sum_{j=0}^{\DP-1} \omega_{\DP}^{ij} \sum_{\substack{\vec{l} \equiv i+j \text{ mod } \DP \\ \vec{l} \in \Rl}} f(\vec{l}) \ket{D_{\lambda^{g\vec{l}j}}^n},
\end{equation}
where $i \in [\DL]$, 
\begin{equation} \label{eqn:vec_l}
    \vec{l} = (l_0,l_1,\dots,l_{\DP-2}) \equiv i+j \text{ mod } \DP \Leftrightarrow l_0 \equiv l_1 \equiv \dots \equiv l_{\DP-2} \equiv i+j \text{ mod } \DP,
\end{equation}
$\Rl$ is the set of all allowed $\vec{l}$, $\lambda^{g\vec{l}j} = (\lambda_0, \lambda_1, \dots,\lambda_{\DP-1})$ are the compositions given by:
\begin{equation}\label{eqn:part_from_lvector}
    \lambda^{g\vec{l}j} = (gl_0,gl_1,\dots,gl_{j-1},n-(\sum_{k=0}^{\DP-2} gl_k ), gl_{j}, \dots, gl_{\DP-2} )
\end{equation}
and $\omega_{\DP} = e^{2\pi \im / \DP}, \; 2<\DP\in \mathbb{N}$ is the $\DP$-th root of unity. 

For this construction to be error-correcting we must have: $n,g,b,t \in \Z ; i,j \in [\DL]$ such that $0 \leq 2t \leq b <n/g$, where:
\begin{equation} \label{eqn:n_constraints}
    n = gb + \delta + 1, \; \delta \geq 2t, g \geq 2t,
\end{equation}
and $b$ is a constant, akin to $m$ from \cref{eqn:Aydin_PI_codes}, determined by the set $\Rl$. Specifically $\Rl$ must satisfy:
\begin{equation}\label{eqn:lmodel_constraint}
    \forall \vec{l},\vec{l'} \in \Rl, \forall k \in [q-1] : l_0 + l_1 + \dots +l_{k-1} + l'_k + l_{k} + \dots + l_{q-2} \leq b.
\end{equation}
Additionally, again by applying condition C3 from the KL equations of \cref{eqn:KL_qudits}, for $\mu \partition{\DP} 2t$, $f(\vec{l})$ must satisfy:
\begin{equation}\label{eqn:C4_cond_Simplicial_codes}
    \sum_{r=0}^{q-1} \sum_{\substack{\vec{l} \equiv r \text{ mod } \DP \\ \vec{l} \in \Rl}} \frac{|f(\vec{l})|^2}{\binom{n}{\lambda^{g\vec{l}r}}} \left[ \binom{n-2t}{\lambda^{g\vec{l}(r-i)} - \mu} - \binom{n-2t}{\lambda^{g\vec{l}(r-j)} - \mu} \right] = 0,
\end{equation}
where $\mu$ is a composition akin to the one seen in the qudit KL equations for PI codes in \cref{thrm:PI_KL_conds_qudits}, i.e.\ it is representing an action of a deletion error, and $i,j \in [\DL]$ with $i>j$ to avoid repeating equivalent equations in the system.
We also require that $f(\vec{l})$ is normalized, from condition C1 of KL equations of \cref{eqn:KL_qudits}:
\begin{equation}\label{eqn:f_vars_norm}
   \sum_{\vec{l} \in \Rl} |f(\vec{l})|^2 =1,
\end{equation}
Because we assume all the $f(\vec{l})$ that fall outside of \cref{eqn:C4_cond_Simplicial_codes} to be zero, the set $\Rl$ can only contain vectors that are entrywise congruent modulo $\DP$. That is, if $\vec{l} =(l_0,l_1,\dots,l_{\DP-2}) \in \Rl$ it must satisfy:
\begin{equation} \label{eqn:l_mod_q}
    l_0\equiv l_1 \equiv \dots \equiv l_{\DP-2} \text{ mod } \DP.
\end{equation}

This code construction requires $\DP \geq \DL$. See \cref{fig:simplex} for a graphical depiction of an example of our construction. We refer to the code construction described in this Section as \textbf{simplicial codes}. Unfortunately, in this construction the scaling $n(t)$ that we observe numerically for $\DL=\DP=3$ is $n(t)=t^3 + 13t^2 + 10t + 1$, which is worse than existing constructions with $n(t) \propto t^2$ \cite{ouyang2017permutation} and suggests that the constraints in \cref{eqn:vec_l,eqn:part_from_lvector,eqn:n_constraints,eqn:lmodel_constraint} are too restrictive. We provide further details on these codes in the Appendix~\ref{ch:Simplicial_codes}, and reasons for why the construction fails to achieve quadratic or subquadratic scaling in \cref{ch:Comparing_simplicial_to_Ruskai}.

\begin{figure}
    \centering
    \includegraphics[width=0.8\textwidth]{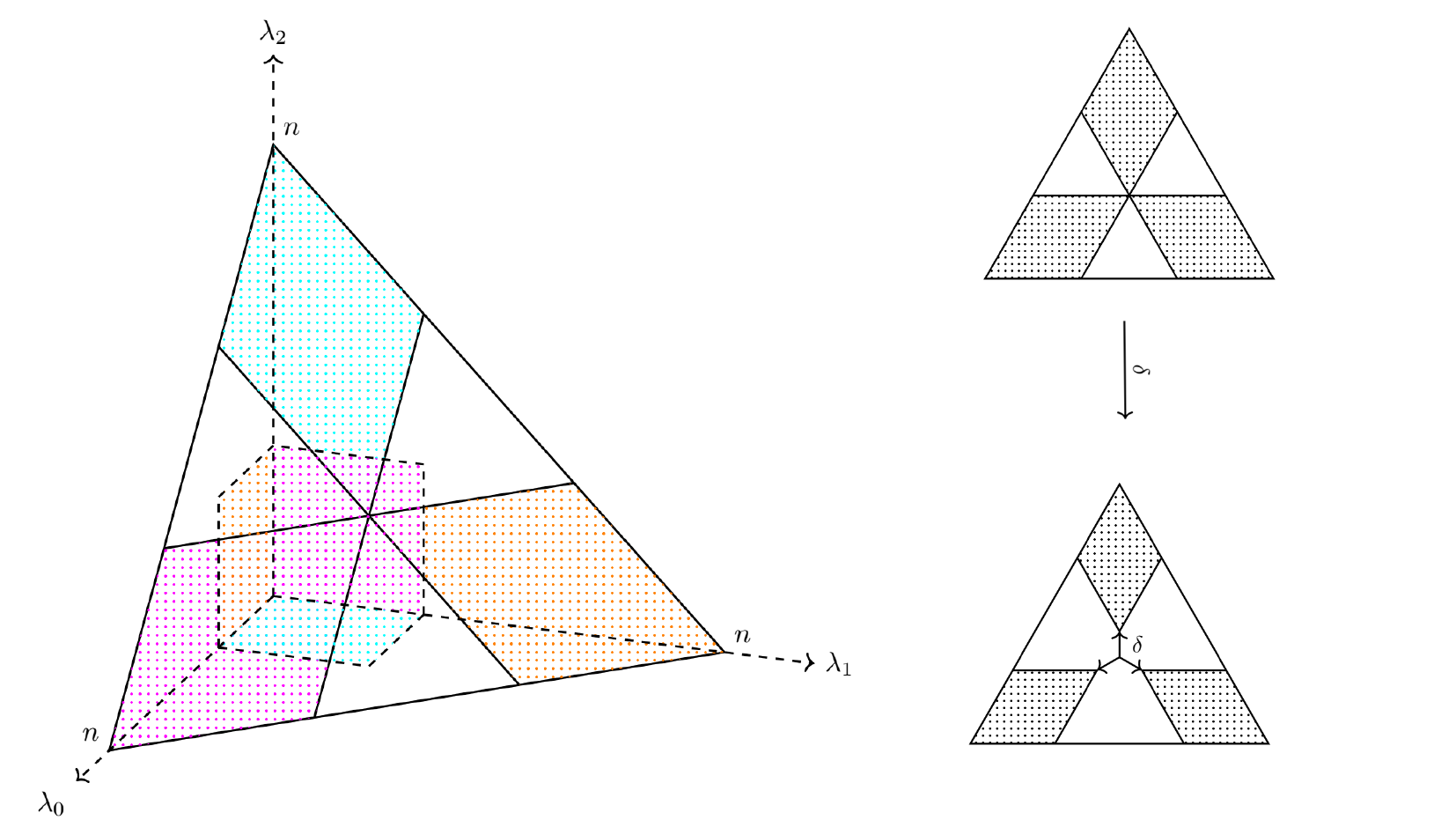}
    \caption{
    \textbf{Left:} graphical depiction of the codespace defined in \cref{sec:Simplicial_codes} 
    constructed from an example set $\Rl := \{ l_i| l_i \leq b/3 \; \forall i \in \{0,1\}\}$ with $\delta = 0$ (for the sake of demonstration as $\delta \geq 2t$). Compositions $\lambda = (\lambda_0,\lambda_1,\lambda_2)$ making up the simplex take on values of one of three types: $(gl_0,gl_1,n-gl_0-gl_1)$, $(gl_0,n-gl_0-gl_1,gl_1)$, and $(n-gl_0-gl_1,gl_0,gl_1)$, and the three colored ``diamonds'' being plots of these three types of compositions, which correspond to the projections of the respective shaded squares defined via $\Rl$. Note that the set $\Rl$ chosen here is suboptimal and is provided purely as an example. See Appendix~\ref{ch:optimal_lmodel} for details on choosing $\Rl$ optimally.
    \textbf{Right:} graphical depiction of the effect of setting $\delta \neq 0$ represents in the code construction, i.e.\ enabling the possibility to correct $t$ errors as $\delta \geq 2t$. This results in an isotropic shift of the "diamonds" starting from the center of the simplex, see Appendix~\ref{sec:Intuition_Simplicial} for more details.
    }
    \label{fig:simplex}
\end{figure}

\section{Discussion}\label{ch:Comparing_simplicial_to_Ruskai}

In our numerical study of qubit PI quantum error-correcting codes, we conjectured a lower bound $n_{\min}(t) \geq 3t^2+3t+1$ (\cref{con:min_PI_qubit_scaling}) on the minimal block length.
This conjecture is based on the assumption that the minimal block length scales with $t$ as a polynomial of integer degree and the values of $n_{\min}(t)$ that we obtained numerically for $t=1$, $t=2$ and $t=3$.
Because we already know that $4t+1 \leq n_{\min}(t) \leq 4t^2 +2t+1$ due to the Singleton bound from below and the AAB scaling from above \cite{aydin2024family}, a proof that $n_{\min}(t)$ is a polynomial of integer degree would significantly strengthen the conjecture. 

Because the block length scaling $n_{\min}(t)$ of minimal PR qubit codes~\cite{pollatsek2004permutationally} is the same as minimal qubit PI codes, this suggests that there may exist a unitary transformation between the two. Such a unitary must leave the KL equations invariant, and therefore must also commute with deletion errors. In the computational basis, one construction would be a tensor product of a unitary acting on the qubits being deleted and a unitary acting on the remaining qubits. Because indexing does not matter for PI codes, without loss of generality we consider the first $2t$ qubits being deleted:
  \begin{equation}
      U = U_{\DP^{2t}} \otimes U_{\DP^{n-2t}},
  \end{equation}
  where $U_{\alpha}$ denotes a unitary acting on a Hilbert space of dimension $\alpha$, and $\DP = 2$.
  At the same time, we know that this unitary should preserve the symmetric subspace, therefore in the Schur basis (see \cite{Bacon_2006}) it should decompose as:
  \begin{equation}
      T_{S} U T^{\dagger}_{S} = U_{d_{\rm sym}} \oplus U_{\DP^n - d_{\rm sym}}.
  \end{equation}
  where $d_{\rm sym}$ is the size of the symmetric subspace, and $T_S$ is the Schur transform. This is equivalent to saying that in the Schur basis: 
  \begin{equation}
      U_{d_{\rm sym}} = U_{2t} \oplus U_{d_{\rm sym} - 2t}.
  \end{equation}
  Using the above insight, one could try to either find such a unitary $U$ or show that it doesn't exist. If it exists, then the two families are equivalent, which may help to simplify analytical proofs of lower bounds for PI codes, as it would be sufficient to provide proofs for the PR code family instead of the general case.

For qudit PI codes, in \cref{sec:num_study_qudit_PI} we showed numerically that for $t=1$, $n_{\min}(t)$ decreases towards the Singleton bound if we fix the logical local dimension $\DL$ and increase the physical local dimension $\DP$. Future work could further investigate this behavior for $t>1$, which we were unable to access in this work due to large computational costs. 

Lastly, there are several possible reasons why our simplicial qudit PI code construction in \cref{eqn:Simplicial_codes}
only achieves block length scaling $O(t^3)$ for $\DP=\DL=3$, which is worse than the $O(t^2)$ scaling that is achieved by padding qubit PI codes, see \cref{sec:converting_qubit_to_qudit}, and by Ouyang qudit PI codes~\cite{ouyang2017permutation}.
Similar to the AAB construction, we require that the codewords have disjoint support, which translates to the requirement that $g\geq2t$ where $g$ is the spacing between points in the simplex. Relaxing this requirement may improve the code performance, i.e.\ one could have $g<2t$ (cf.\ the proof of condition C2 in Appendix~\ref{sec:Simplicial_code_proofs}, the $r=s$ case).
One could also try to see what is the smallest value $g$ could take in terms of $t$, so that in C1 codewords still have disjoint support, while in C2 they satisfy orthogonality without imposing disjoint support of erroneous codewords under the action of deletion errors. The idea of using the same scaling parameter $g$ for all elements of the composition might itself be too restrictive. Introducing more scaling parameters could lead to a PI code family that can achieve better scaling compared to Ouyang \cite{ouyang2017permutation} or AAB \cite{aydin2025quantumerrorcorrectionsu2}. Finally, note that it is unclear whether increasing the physical local dimension $\DP$ in the construction of simplicial codes would help. It is possible that perhaps for higher $\DP$ simplicial codes can achieve subquadratic scaling. To answer this, one potentially useful direction to investigate would be to check whether it is possible to make simplicial codes satisfy the conditions given in the construction from Section VII of \cite{aydin2025quantumerrorcorrectionsu2}, as the two constructions appear to share some underlying ideas. 
Even though our construction gives suboptimal scaling for $\DP=\DL=3$, we believe it can offer insight on how to construct explicit qudit PI codes that achieve subquadratic scaling of code distance with the block length.

\section{Acknowledgments}
We thank Victor Albert, Marek Miller, Nicolas Resch, and Gavin Brennen for insightful discussions. This work was supported by the NWO grant NGF.1623.23.025 (``Qudits in theory and experiment''). A.S.N. is supported by the Dutch Research Council (NWO/OCW) as a part of the Quantum Software Consortium (project number 024.003.037), Quantum Delta NL (project number NGF.1582.22.030) and ENWXL grant (project number OCENW.XL21.XL21.122). We thank SURF for the support in using the Dutch National Supercomputer Snellius. 

\clearpage

\printbibliography

\clearpage

\markboth{Appendix}{Appendix}

\appendix

\section{Minor proofs}\label{apx:proofs}

\noindent \textbf{Proof of \cref{lem:Kraus_op_decomp}:}
\begin{proof}
We have:
\begin{align}
        &A_{i_1,x_1}^{n-t+1} \circ A_{i_2,x_2}^{n-t+2} \circ \dots \circ A_{i_t,x_t}^{n} \\ 
        &= A_{i_1,x_1}^{n-t+1} \circ A_{i_2,x_2}^{n-t+2} \circ \dots \circ \left(\bra{x_t}_{i_t} \otimes \id_{\{1,\dotsc,n\} \setminus \{i_t\}}\right) \\
        &= \left(\bra{x_1}_{i_1} \otimes \id_{\{1,\dotsc,n\} \setminus \{i_1,i_2,\dots,i_t\}}\right) \circ \left(\bra{x_2}_{i_2} \otimes \id_{\{1,\dotsc,n\} \setminus \{i_2,\dots,i_t\}}\right) \circ \dots \circ \left(\bra{x_t}_{i_t} \otimes \id_{\{1,\dotsc,n\} \setminus \{i_t\}}\right) \\
        &= \bra{x}_I \otimes \id_{\{1,\dots,n\} \setminus I} = A_{I,x}^n. 
    \end{align}
    \qedhere
\end{proof}

\noindent \textbf{Proof of \cref{lem:del_err_Dicke}:}
\begin{proof}
    We have:
    \begin{align}
        E_{\mu}^q \ket{D_{\l}^n} &=  (\Chi_{q-1})^{\mu_{q-1}} \circ \cdots \circ (\Chi_{0})^{\mu_{0}} \ket{D_{\l}^n} \\
        &= \sqrt{\frac{\binom{n-\mu_0-\mu_1-\dots-\mu_{q-1}}{(\l_0-\mu_0, \l_1-\mu_1, \dots, \l_{q-1}-\mu_{q-1})}}{\binom{n}{\l}}} \ket{D_{(\l_0 - \mu_0, \dots, \l_{q-1} - \mu_{q-1})}^{n-a}} \\
        &= \sqrt{\frac{\binom{n-t}{\l - \mu}}{\binom{n}{\l}}} \ket{D_{\l - \mu}^{n-t}}. 
    \end{align}
        \qedhere
\end{proof}

\section{Numerical results for coefficient fixing in qubit codes with $t=2$ and $t=3$}\label{paramfixing_t2t3}
In this appendix we show additional numerical results when $M$ codeword coefficients are fixed for the case of real minimal PI codes with the mirror and phase-flip symmetry of \cref{sec:codewordsymmetries} enforced. We obtain solutions by numerically by minimizing the cost function of \cref{eq:f_cost} using gradient descent methods as described in \cref{ch:Numerical_study}. \cref{fig:ParamFixingSymmetryFix}(a) shows $(M,t) = (3,2)$ with $\alpha_0,\alpha_1,\alpha_2 = 0 $, and \cref{fig:ParamFixingSymmetryFix}(b) shows $(M,t) = (4,3)$ with $\alpha_0,\alpha_1,\alpha_2,\alpha_3 = 0$. The numerical minimization starts from $20\times10^4$ initial random guesses for the unfixed parameters for $(M,t) = (3,2)$ and $5.6\times10^4$ initial random guesses for $(M,t) = (4,3)$. We obtain only eight solutions for $(M,t) = (3,2)$ and and $16$ solutions for $(M,t) = (4,3)$. This numerical evidence forms the basis of \cref{con:algebraic_variety}. 

\begin{figure}
    \centering
    \includegraphics{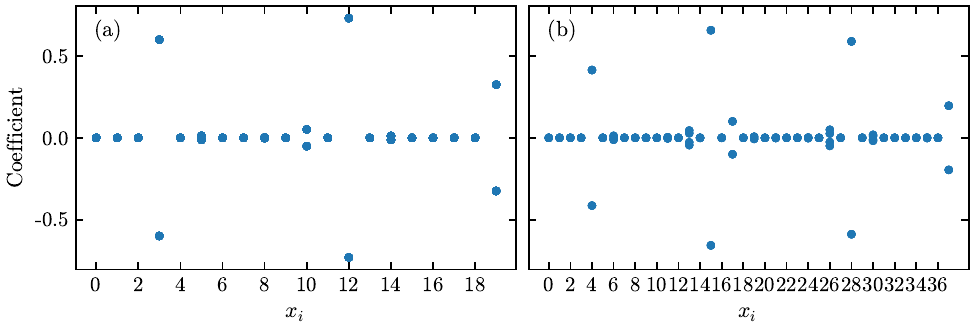}
    \caption{Codeword coefficients $x_i = (\alpha_0,\dots,\alpha_n,\beta_0,\dots,\beta_n)$ of real minimal qubit PI codes with mirror and phase-flip symmetries enforced, for the cases $(M,t) = (3,2)$ in panel (a) and $(M,t) = (4,3)$ in panel (b). The $M = 3$ fixed coefficients in (a) are $\alpha_0,\alpha_1,\alpha_2 = 0$ while the $M = 4$ fixed coefficients in (b) are $\alpha_0,\alpha_1,\alpha_2,\alpha_3 = 0$. We find only eight solutions in (a) and $16$ solutions in (b).  
    }
    \label{fig:ParamFixingSymmetryFix}
\end{figure}

\section{Qudit PI codes on discrete simplices} 
\label{ch:Simplicial_codes}

\subsection{AAB code construction}\label{sec:Aydin_codes}

Let us recall the code construction in \cite{aydin2024family}:

\begin{equation}\label{eqn:Aydin_PI_codes}
\begin{split}
    &\ket{c_0} = \sum_{\substack{l \text{ even} \\ 0 \leq l \leq m}} f(l) \ket{D_{gl}^n} + \sum_{\substack{l \text{ odd} \\ 0 \leq l \leq m}} f(l) \ket{D_{n-gl}^n}, \\
    &\ket{c_1} = \sum_{\substack{l \text{ odd} \\ 0 \leq l \leq m}} f(l) \ket{D_{gl}^n} + \epsilon\sum_{\substack{l \text{ even} \\ 0 \leq l \leq m}} f(l) \ket{D_{n-gl}^n}.
\end{split}
\end{equation}
Here $\epsilon$ is an orthogonality coefficient $\epsilon = \pm 1$ and $n,g,m,t \in \N: 0\leq 2t \leq 2m <n/g$.
The parameters $m$ and $g$ specify the codes - $m$ defines how many Dicke states constitute a codeword, and for the purpose of $g$ see Appendix~\ref{sec:Intuition_Aydin} below. 
The KL conditions impose constraints on the block length $n$. To be error-correcting we must have:
\begin{equation}
    n = 2gm + \delta + 1, \; \delta \geq 2t, (g \geq 2t, \epsilon = -1) \text{ or } (g \geq 2t + 1, \epsilon = +1),
\end{equation}
where $\delta$ is a shift parameter, cf.\ \cref{sec:Intuition_Simplicial}. For full derivation see \cite{aydin2024family}. Here $f(l)$ are the codeword coefficients, and they must also satisfy the following conditions derived when acting with KL equations directly on this code construction:
\begin{equation}
    \sum_{l=0}^m (-1)^l \frac{f(l)^2}{\binom{n}{gl}} \left[ \binom{n-2t}{gl-a} - \binom{n-2t}{gl-2t+a} \right] =0,
\end{equation}
which follows from C3 of the KL equations of \cref{eqn:KL_PI_conds_qubits}. The coefficients should also be normalized (cf. C1 in \cref{eqn:KL_PI_conds_qubits}):
\begin{equation}
    \sum_{l=0}^m f(l)^2 = 1.
\end{equation}

Aydin et. al. make the following (non-unique) choice for the codeword coefficients:
\begin{equation}\label{eqn:Aydin_f(l)}
    f(l) = \sqrt{\binom{n/(2g)}{m} \frac{n-2gm}{g(m+1)} \frac{\binom{m}{l}}{\binom{n/g-l}{m+1}}}.
\end{equation}

\subsection{A perspective on the AAB code construction}\label{sec:Intuition_Aydin}
A useful graphical interpretation of the AAB code family is representing the basis vectors in the support of the codespace, i.e.\ the Dicke states, as points of a 1-dimensional discrete simplex. For $\DP=\DL=2$ each Dicke state is in one-to-one correspondence to a composition of $n$ into $\DP=2$ parts. We can therefore consider a simplex $\pazocal{S}=\{\lambda=(\lambda_0,\lambda_1) \in \N^2|\lambda_0 +\lambda_1 = n\}$, with the discrete set of points $(gl, n-gl) \in \pazocal{S}$ and $(n-gl, gl) \in \pazocal{S}$ corresponding to the Dicke states on which the codewords in \cref{eqn:Aydin_PI_codes} are supported, for $l \in \{0,1,\dots,m\}$. This simplex is shown in \cref{fig:simplex_q=2}. Everywhere, we refer to $g$ as the scaling parameter. Intuitively, the purpose of $g$ is to increase the spacing between the discrete set of points in the simplex to account for the shifts caused by deletion errors, preventing overlap between the Dicke states after their weight is shifted by a deletion error.
\\
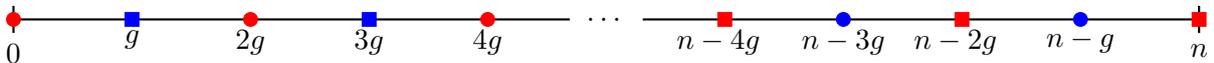
\begin{figure}[h]
    \centering
    \begin{tikzpicture}[scale=1.2, thick, every node/.style={scale=1.1}]
    
        \def\n{13}
        \def\s{1.3}
        
        \draw[-] (0,0) -- (\n/2-0.4,0);
        \draw[-] (\n/2+0.4,0) -- (\n,0);
    
        \draw[thick] (0,0.15) -- (0,-0.15);
        \node[below] at (0,-0.15) {$0$};
    
        \draw[thick] (\n,0.15) -- (\n,-0.15);
        \node[below] at (\n,-0.15) {$n$};

        \filldraw[red] (0,0) circle (2pt);

        \def\sq{0.07}

        \filldraw[blue] (\s-\sq,\sq) rectangle ++(2*\sq,-2*\sq);
        \node[below] at (\s,0) {$g$};
        \filldraw[blue] (3*\s-\sq,\sq) rectangle ++(2*\sq,-2*\sq);
        \node[below] at (3*\s,0) {$3g$};
    
        \foreach \i in {2,4} {
          \filldraw[red] (\i*\s,0) circle (2pt);
          \node[below] at (\i*\s,0) {$\i g$};
        }

        \filldraw[red] (\n-\sq,\sq) rectangle ++(2*\sq,-2*\sq);

        \filldraw[blue] (\n-\s,0) circle (2pt);
        \node[below] at (\n-\s,0) {$n-g$};
        \filldraw[blue] (\n-3*\s,0) circle (2pt);
        \node[below] at (\n-3*\s,0) {$n-3g$};
        
        \foreach \j in {4,2} {
          \pgfmathsetmacro\x{\n - \j*\s - \sq}
          \filldraw[red] (\x,\sq) rectangle ++(2*\sq,-2*\sq);
          \node[below] at (\x,0) {$n-\j g$};
        }
    
        \node at (\n/2,0) {\dots};
    
    \end{tikzpicture}
    \caption{Simplex $\pazocal{S}$ as a graphical interpretation of the codespace of the AAB PI code family. Here the points denote the entry $\lambda_1$ of the Dicke state compositions, which is sufficient since $\lambda_0$ can always be recovered as $\lambda_0 = n - \lambda_1$. Circles $\bullet$ correspond to terms in the support of the $c_0$ codeword, and squares \scalebox{0.6}{$\blacksquare$} to terms that are in the support of the $c_1$ codeword. $\textcolor{red}{Red}$ shapes correspond to terms with $\textcolor{red}{even}$ $l$, and $\textcolor{blue}{blue}$ to terms with $\textcolor{blue}{odd}$ $l$.
    }
    \label{fig:simplex_q=2}
\end{figure}

\subsection{Intuition behind the construction of simplicial codes}\label{sec:Intuition_Simplicial}
Recall the construction of qudit PI codes in \cref{sec:Simplicial_codes}. Similarly to the qubit construction, $n$ is a composition into $\DP$ parts with $\DP > 2$ and the codespace can be represented as a simplex of dimension $\DP-1$.
To explain the construction, let us first consider a concrete example $\DP=\DL=3$, where we choose the region $\Rl$ as: $\Rl := \{ l_i| l_i \leq b/3 \; \forall i \in \{0,1\}\}$ for some $b$. We also set $\delta = 0$ for simplicity. Note that this is not a valid error-correcting code since $\delta \geq 2t$; we return to the case of non-trivial $\delta$ below. The region $\Rl$ in the $(l_0,l_1)$ plane is depicted in \cref{fig:Mitsubishi_l_model}.

\begin{figure}
  \begin{minipage}{0.49\textwidth}
  \centering
  \begin{tikzpicture}[scale=1]
    \def\a{6}
    \draw[thick, -{Latex}] (0,0) -- ({\a+0.5},0)   node[below] {$l_0$};
    \draw[thick,-{Latex}] (0,0) -- (0,{\a+0.5}) node[left]    {$l_1$};
  
    \coordinate (A) at (0, 0);
    \coordinate (B) at (\a, 0);
    \coordinate (C) at (0, \a); 
    \coordinate (E1) at (0,\a/3);
    \coordinate (E2) at (0,2*\a/3);
    
    \coordinate (F1) at (\a/3,2*\a/3 );
    \coordinate (F2) at (2*\a/3,\a/3 );
    
    \coordinate (G1) at (\a/3, 0);
    \coordinate (G2) at (2*\a/3, 0);

    \coordinate (O) at ({\a/3}, {\a/3});
      
    \draw[thick] (C) -- (B);
    \draw[thick] (G1) -- (O);
    
    \draw[thick] (E1) -- (O);
    
    \fill[pattern={Lines[angle=45, distance=4pt, line    width=0.5pt]}] (A) -- (E1)-- (O) --(G1)-- cycle;
  
    \draw[thick] (\a/3,0.1) -- (\a/3,-0.1);
    \draw[thick] (-0.1,\a/3) -- (0.1,\a/3);
    \node at (E1) [left]  {$b/3$};
    \node at (G1) [below]  {$b/3$};
    \node at (C) [left]  {$b$};
    \node at (B) [below]  {$b$};
  \end{tikzpicture}
  \caption{Graphical depiction of the set $\Rl := \{ l_i| l_i \leq b/3 \; \forall i \in \{0,1\}\}$ as the shaded region of the figure.}
  \label{fig:Mitsubishi_l_model}
  \end{minipage}
  \hfill
  \begin{minipage}{0.49\textwidth}
    \centering
    \begin{tikzpicture}[scale=1.2]
     \def\a{6}
     \coordinate (A) at (0, 0);
     \coordinate (B) at (\a, 0);
     \coordinate (C) at (\a/2, {\a/2 * sqrt(3)}); 
     \coordinate (E1) at (\a/6,{\a/2 * sqrt(3)/3});
     \coordinate (E2) at (\a/3,{\a/2 * 2* sqrt(3)/3});
  
     \coordinate (F1) at ({2*(\a)/3},{\a/2 *2 * sqrt(3) / 3} );
     \coordinate (F2) at ({5/6*\a},{\a/2 * sqrt(3) / 3} );
  
     \coordinate (G1) at (\a/3, 0);
     \coordinate (G2) at (2*\a/3, 0);

     \coordinate (O) at (\a/2, {1/(2*sqrt(3))*\a});
  
     \draw[thick] (A) -- (B) -- (C) -- cycle;
  
     \draw[thick] (E1) -- (F2);
     \draw[thick] (E2) -- (G2);
     \draw[thick] (F1) -- (G1);

     \fill[pattern=dots] (A) -- (E1) -- (O) -- (G1)-- cycle;
     \fill[pattern=dots] (C) -- (E2) -- (O) -- (F1)-- cycle;
     \fill[pattern=dots] (B) -- (F2) -- (O) -- (G2)-- cycle;
    \end{tikzpicture}
    \caption{Codespace formed from the set $\Rl$ in \cref{fig:Mitsubishi_l_model} for $\DP=3$, shown as the dotted subset of the simplex.}
    \label{fig:Mitsubishi_2D}
  \end{minipage}
\end{figure}

To construct the full codespace as a set of points of the form $(gl_0,gl_1,n-gl_0-gl_1)$, $(gl_0,n-gl_0-gl_1,gl_1)$, and $(n-gl_0-gl_1,gl_0,gl_1)$ on the simplex $\pazocal{S}=\{\lambda=(\lambda_0,\lambda_1,\lambda_2) \in \R^3|\lambda_0 +\lambda_1+\lambda_2 = n\}$ we notice that those three  compositions of points correspond to the shaded part of the region $\Rl$ on three different coordinate projections: to $(\lambda_0,\lambda_1)$, $(\lambda_0,\lambda_2)$, and $(\lambda_1,\lambda_2)$ planes, respectively. From these projections we can reconstruct the full codespace, as shown in \cref{fig:Mitsubishi_3D}.

We then have a simplex of dimension two, with the discrete set of points on it as described above, and with three vertices corresponding to the cases $\lambda = (n,0,0)$, $\lambda = (0,n,0)$, and $\lambda = (0,0,n)$ respectively, denoted by $j=0, j=1, j=2,$ as described in the code construction in \cref{sec:Simplicial_codes}. Similarly to the AAB code construction, see \cref{sec:Intuition_Aydin}, the main idea is that different codewords have disjoint support on the Dicke states. This is shown in \cref{fig:Mitsubishi_patterns}.

Although in this example we have until now considered $\delta=0$, our construction requires $\delta \geq 2t$. The effect of non-zero $\delta$ on the shape of the allowed compositions $\vec{l}$ as defined by the relations \cref{eqn:lmodel_constraint,eqn:n_constraints} is demonstrated in \cref{fig:delta_graphical}. This can be understood as shifting the subsets of the simplex, creating more space to ensure no overlap occurs between the different "diamonds" constituting the region $\Rl$ when introducing deletion errors.

\begin{figure}
  \centering
  \includegraphics{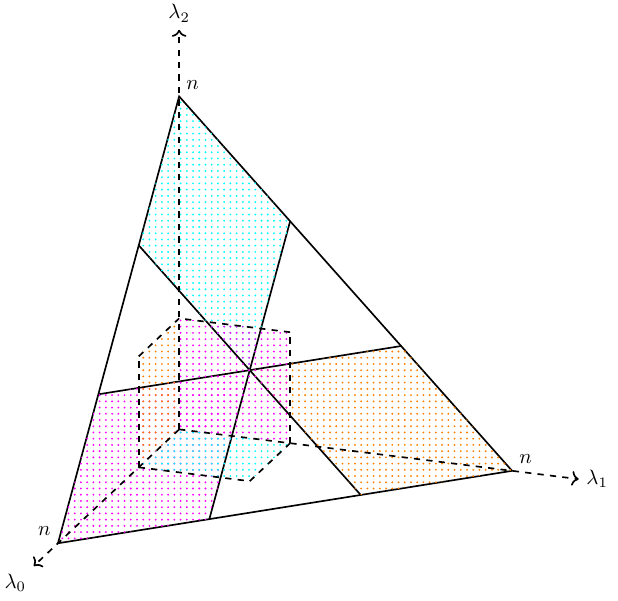}
  \caption{Graphical depiction of the codespace defined in \cref{sec:Simplicial_codes}, constructed from the set $\Rl := \{ l_i| l_i \leq b/3 \; \forall i \in \{0,1\}\}$ depicted in \cref{fig:Mitsubishi_l_model} with $\delta = 0$. Colors of the subsets $(gl_0,gl_1,n-gl_0-gl_1)$, $(gl_0,n-gl_0-gl_1,gl_1)$, and $(n-gl_0-gl_1,gl_0,gl_1)$ correspond to their respective projections (shaded squares) defined via $\Rl$.}
  \label{fig:Mitsubishi_3D}
\end{figure}

\begin{figure}
  \centering
  \begin{minipage}{0.32\textwidth}
    \centering
    \includegraphics{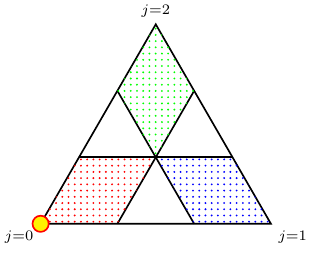}
    \caption*{$c_0$ codeword}
  \end{minipage}
  \hfill
  \begin{minipage}{0.32\textwidth}
    \centering
    \includegraphics{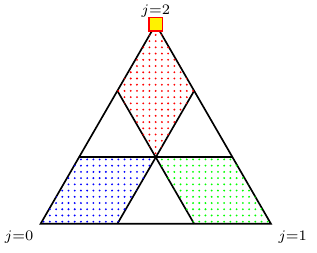}
    \caption*{$c_1$ codeword}
  \end{minipage}
  \hfill
  \begin{minipage}{0.32\textwidth}
    \centering
    \includegraphics{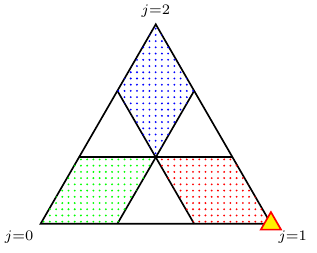}
    \caption*{$c_2$ codeword}
  \end{minipage}

  \caption{Graphical interpretation of our simplicial code construction for qutrit codes, $\DP=\DL=3$. The region $\Rl$ for this example is described in \cref{sec:Intuition_Simplicial} and depicted in \cref{fig:Mitsubishi_l_model}. Dots corresponds to some unique composition $\lambda=(\lambda_0,\lambda_1,\lambda_2)$ such that no dot overlaps with another one. In order to not clutter the image, we divided it into 3 parts, each corresponding to their respective codeword support. $\textcolor{red}{Red}$ shapes correspond to terms with $\textcolor{red}{\vec{l} \equiv 0 \mod \DP}$, $\textcolor{blue}{blue}$ to terms with $\textcolor{blue}{{\vec{l} \equiv 1 \mod \DP}}$, and $\textcolor{green}{green}$ to terms with $\textcolor{green}{{\vec{l} \equiv 2 \mod \DP}}$, cf. the relation of \cref{eqn:l_mod_q}. Besides the other points of the simplex, only one of the three vertices belongs to each codeword: the yellow circle, square, and triangle denote which of the vertices of the simplex belong to the $c_0$, $c_1$, and $c_2$ codewords respectively.}
  \label{fig:Mitsubishi_patterns}
\end{figure}

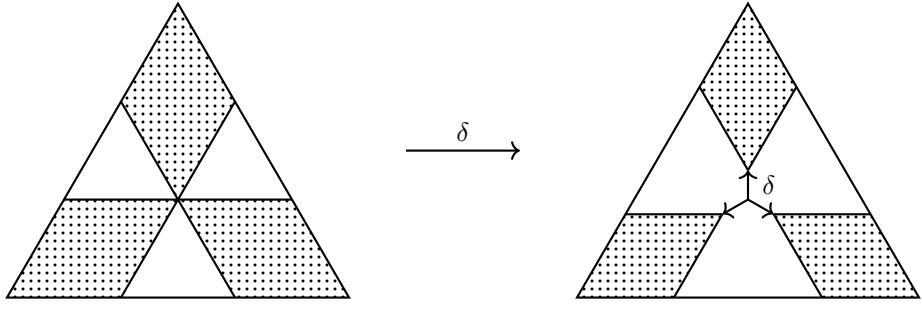
\begin{figure}
\centering
\begin{tikzpicture}[scale=1.5, thick]
  \pgfmathsetmacro{\a}{3}
  \pgfmathsetmacro{\h}{\a/2 * sqrt(3)}
  \pgfmathsetmacro{\s}{0.85}
\begin{scope}
  \coordinate (A) at (0, 0);
  \coordinate (B) at (\a, 0);
  \coordinate (C) at (\a/2, {\a/2 * sqrt(3)}); 
  \coordinate (E1) at (\a/6,{\a/2 * sqrt(3)/3});
  \coordinate (E2) at (\a/3,{\a/2 * 2* sqrt(3)/3});
  
  \coordinate (F1) at ({2*(\a)/3},{\a/2 *2 * sqrt(3) / 3} );
  \coordinate (F2) at ({5/6*\a},{\a/2 * sqrt(3) / 3} );
  
  \coordinate (G1) at (\a/3, 0);
  \coordinate (G2) at (2*\a/3, 0);

  \coordinate (O) at (\a/2, {1/(2*sqrt(3))*\a});
  
  \draw[thick] (A) -- (B) -- (C) -- cycle;
  
  \draw[thick] (E1) -- (F2);
  \draw[thick] (E2) -- (G2);
  \draw[thick] (F1) -- (G1);

  \fill[pattern=dots] (A) -- (E1) -- (O) -- (G1)-- cycle;
  \fill[pattern=dots] (C) -- (E2) -- (O) -- (F1)-- cycle;
  \fill[pattern=dots] (B) -- (F2) -- (O) -- (G2)-- cycle;
\end{scope}

  \draw[->, thick] (\a+0.5,\h/2) -- node[above] {$\delta$} (\a+1.5,\h/2);

\begin{scope}[shift={(\a+2,0)}]
  \coordinate (A) at (0, 0);
  \coordinate (B) at (\a, 0);
  \coordinate (C) at (\a/2, {\a/2 * sqrt(3)}); 
  \coordinate (E1) at (\a/6*\s,{\a/2 * sqrt(3)/3 * \s});
  \coordinate (E2) at (\a/2-\a/6*\s,\h - \h/3 * \s);
  
  \coordinate (F1) at ({\a/2+\a/6*\s},{\h - \h/3 * \s} );
  \coordinate (F2) at ({\a-\a/6*\s},{\a/2 * sqrt(3) / 3 * \s} );
  
  \coordinate (G1) at (\a/3*\s, 0);
  \coordinate (G2) at (\a-\a/3*\s, 0);

  \coordinate (O) at (\a/2, {1/(2*sqrt(3))*\a});
  \coordinate (OA) at (\a/2*\s, {1/(2*sqrt(3))*\a*\s});
  \coordinate (OB) at (\a-\a/2*\s, {1/(2*sqrt(3))*\a*\s});
  \coordinate (OC) at (\a/2, {\h-1/(sqrt(3))*\a*\s)});

  \draw[thick] (A) -- (B) -- (C) -- cycle;

  \draw[<-] (OC) -- (O) node[midway, right]{$\,\delta$};
  \draw[<-] (OA) -- (O);
  \draw[<-] (OB) -- (O);
  
  \draw[thick] (E1) -- (OA);
  \draw[thick] (OA) -- (G1);
  \draw[thick] (E2) -- (OC);
  \draw[thick] (OC) -- (F1);
  \draw[thick] (G2) -- (OB);
  \draw[thick] (OB) -- (F2);

  \fill[pattern=dots] (A) -- (E1) -- (OA) -- (G1)-- cycle;
  \fill[pattern=dots] (C) -- (E2) -- (OC) -- (F1)-- cycle;
  \fill[pattern=dots] (B) -- (F2) -- (OB) -- (G2)-- cycle;
\end{scope}
\end{tikzpicture}
\caption{Action of adding the shift $\delta$, illustrated graphically (shift happens from the center).
}
\label{fig:delta_graphical}
\end{figure}

\subsection{Proof that simplicial codes \cref{eqn:Simplicial_codes} are error-correcting}\label{sec:Simplicial_code_proofs}

In this section we will prove that the simplicial PI codes are error correcting, which is formulated in the following Theorem:
\begin{theorem}
    The qudit PI codes described in \cref{sec:Simplicial_codes}, referred to as simplicial codes, satisfy the KL conditions of \cref{eqn:KL_qudits}.
\end{theorem}
\begin{proof}\
\subsubsection*{$\boxed{\mathrm{\textbf{C1}}}$}
By definition of simplicial codes given in \cref{eqn:Simplicial_codes}, we have:
\begin{equation}\label{eqn:Simplicial_code_coefs}
    \alpha_{i,\lambda} = \sum_{j=0}^{\DP-1} \omega_{\DP}^{ij} \sum_{\substack{\vec{l} \equiv i+j \text{ mod } \DP \\ \vec{l} \in \Rl}} f(\vec{l}) \delta_{\lambda, \lambda^{g\vec{l}j}}.
\end{equation}
First, it is easy to see that within our construction the codewords are normalized, since $\forall i \in [\DL]$:
\begin{align}
     \sum_{\l \partition{q} n} \bar{\alpha}_{i, \l} \alpha_{i, \l} &= \sum_{\l \partition{q} n}\sum_{j,j'=0}^{{\DP}-1} \omega_{\DP}^{-ij}\omega_{\DP}^{ij'} \sum_{\substack{\vec{l} \equiv i+j \text{ mod } {\DP} \\ \vec{l} \in \Rl}} \sum_{\substack{\vec{l'} \equiv i+j' \text{ mod } {\DP} \\ \vec{l'} \in \Rl}}\bar{f}(\vec{l})f(\vec{l'})\delta_{\lambda,\lambda^{g\vec{l}j}}\delta_{\lambda, \lambda^{g\vec{l'}j'}} \\
     &= \sum_{\l \partition{q} n}\sum_{j,j'=0}^{{\DP}-1} \omega_{\DP}^{-ij}\omega_{\DP}^{ij'} \sum_{\substack{\vec{l} \equiv i+j \text{ mod } {\DP} \\ \vec{l} \in \Rl}} \sum_{\substack{\vec{l'} \equiv i+j' \text{ mod } {\DP} \\ \vec{l'} \in \Rl}}\bar{f}(\vec{l})f(\vec{l'})\delta_{\lambda,\lambda^{g\vec{l}j}}\delta_{\lambda^{g\vec{l}j}, \lambda^{g\vec{l'}j'}} \delta_{j,j'} \delta_{\vec{l},\vec{l'}} \\
    &=\sum_{j=0}^{{\DP}-1}\sum_{\substack{\vec{l} \equiv i+j \text{ mod } {\DP} \\ \vec{l} \in \Rl}} \bar{f}(\vec{l})f(\vec{l}) = \sum_{\vec{l} \in \Rl} |f(\vec{l})|^2 = 1,
\end{align}
where the last equality holds by assumption from \cref{eqn:f_vars_norm}.

Next, just like in AAB family of PI codes, to simplify solving the equations in \cref{thrm:PI_KL_conds_qudits}, our intuition when constructing a family of qudit PI codes was to impose a stronger constraint on codewords than orthogonality, namely to make the codewords have disjoint support in the symmetric subspace. That is, our construction has the following property:
\begin{equation}\label{eqn:Disjoint_support}
     \bar{\alpha}_{i,\l} \alpha_{j,\l} = 0, \; \forall i\neq j \in [\DL],
\end{equation}
from which $ \sum_{\l \partition{\DP} n} \bar{\alpha}_{i, \l} \alpha_{j, \l} = 0, \; \forall i\neq j \in [\DL]$ would trivially follow. To show \cref{eqn:Disjoint_support}, we evaluate:
\begin{eqnarray}
    \bar{\alpha}_{i, \l} \alpha_{j, \l}
     &=& \sum_{r,s=0}^{\DP-1} \omega_{\DP}^{-ir}\omega_{\DP}^{js} \sum_{\substack{\vec{l} \equiv i+r \text{ mod } \DP \\ \vec{l} \in \Rl}} \sum_{\substack{\vec{l'} \equiv j+s \text{ mod } \DP \\ \vec{l'} \in \Rl}}\bar{f}(\vec{l})f(\vec{l'})\delta_{\lambda,\lambda^{g\vec{l}r}}\delta_{\lambda^{g\vec{l}r}, \lambda^{g\vec{l'}s}}.
\end{eqnarray}
$\bullet$ For terms with $r=s$:
\begin{equation}
    \delta_{\lambda^{g\vec{l}r}, \lambda^{g\vec{l'}s}} =\delta_{\vec{l}, \vec{l'}},
\end{equation}
and since for these terms
\begin{equation}
\begin{rcases}
\begin{split}
    &l_0\equiv l_1 \equiv \dots \equiv l_{\DP-2} \equiv i + r \text{ mod } \DP \\
    &l_0'\equiv l_1' \equiv \dots \equiv l_{\DP-2}' \equiv j + r \text{ mod } \DP
\end{split} \;
\end{rcases}, i\neq j \Rightarrow \vec{l} \neq \vec{l'},
\end{equation}
thus $\delta_{\lambda^{g\vec{l}r}, \lambda^{g\vec{l'}r}} =\delta_{\vec{l}, \vec{l'}} = 0$.
\\
\\
$\bullet$ For terms with $r \neq s$:
\begin{equation}
     \delta_{\lambda^{g\vec{l}r}, \lambda^{g\vec{l'}s}} = \delta_{l_0,l_0'}\delta_{l_1,l_1'} \dots \delta_{l_{r-1},l_{r-1}'}\delta_{(n-S_{g\vec{l}}),gl_r'}\delta_{l_{r},l_{r+1}'} \dots \delta_{l_{s-2},l_{s-1}'}\delta_{gl_{s-1},(n-S_{g\vec{l'}})}\delta_{l_s,l_s'} \dots \delta_{l_{\DP-2},l_{\DP-2}'},
\end{equation}
but due to \cref{eqn:lmodel_constraint}, we have
\begin{equation}
   gl_r'+S_{g\vec{l}} \leq gb =n-\delta -1 < n \Rightarrow \delta_{(n-S_{g\vec{l}}),gl_r'} = 0,
\end{equation}
thus
$
    \delta_{\lambda^{g\vec{l}r}, \lambda^{g\vec{l'}s}} = 0
$, 
which concludes the proof of C1. \hfill $\square$

\subsubsection*{$\boxed{\mathrm{\textbf{C2}}}$}
Consider some $\mu, \nu \partition{\DP} 2t$. In the following, we show that for all $\lambda,\mu,\nu$ and all $i \neq j$ we have $\bar{\alpha}_{i,\lambda} \alpha_{j,\lambda - \mu + \nu} = 0$, which is a stronger condition than required for C2 to hold. Without loss of generality, here we can consider only those $\lambda$ that satisfy $\lambda - \mu \geq 0$, since due to proof of C2 of \cref{thrm:PI_KL_conds_qudits} the terms with $\lambda$ that don't satisfy this are zero. Then we have for all $i \neq j$:
\begin{equation}
    \bar{\alpha}_{i,\lambda} \alpha_{j,\lambda - \mu + \nu} = \sum_{r,s=0}^{\DP-1} \omega_{\DP}^{-ir}\omega_{\DP}^{js} \sum_{\substack{\vec{l} \equiv i+r \text{ mod } \DP \\ \vec{l} \in \Rl}} \sum_{\substack{\vec{l'} \equiv j+s \text{ mod } \DP \\ \vec{l'} \in \Rl}}\bar{f}(\vec{l})f(\vec{l'})\delta_{\lambda,\lambda^{g\vec{l}r}}\delta_{\lambda^{g\vec{l}r} - \mu + \nu, \lambda^{g\vec{l'}s}}.
\end{equation} 
$\bullet$ For terms with $r \neq s$:
\begin{align}\begin{split}
     \delta_{\lambda^{g\vec{l}r} - \mu + \nu, \lambda^{g\vec{l'}s}} = & \delta_{gl_0-\mu_0+\nu_0,gl_0'}\delta_{gl_1-\mu_1+\nu_1,gl_1'} \dots \delta_{gl_{r-1}-\mu_{r-1}+\nu_{r-1},gl_{r-1}'} \\ 
     &\times \delta_{(n-\mu_r+\nu_r-S_{g\vec{l}}),gl_r'}\delta_{gl_{r}-\mu_{r+1}+\nu_{r+1},gl_{r+1}'} \dots \delta_{gl_{s-2}-\mu_{s-1}+\nu_{s-1},gl_{s-1}'} \\ 
     &\times \delta_{gl_{s-1}-\mu_s+\nu_s,(n-S_{g\vec{l'}})} \delta_{gl_s-\mu_{s+1}+\nu_{s+1},gl_s'} \dots \delta_{gl_{\DP-2}-\mu_{\DP-1}+\nu_{\DP-1},gl_{\DP-2}'},
\end{split}\end{align}
but due to \cref{eqn:lmodel_constraint} we have:
\begin{align}
   gl_r'+\mu_r-\nu_r+S_{g\vec{l}} \leq gb + \mu_r - \nu_r = n + \mu_r - \nu_r - \delta -1 \leq n + 2t - \delta - 1 \leq n-1 < n,
\end{align}
which implies that $\delta_{(n-\mu_r+\nu_r-S_{g\vec{l}}),gl_r'} = 0$ and therefore $\delta_{\lambda^{g\vec{l}r}-\mu+\nu, \lambda^{g\vec{l'}s}} = 0$.
\\
\\
$\bullet$ For terms with $r = s$:
\begin{equation}
     \delta_{\lambda^{g\vec{l}r} - \mu + \nu, \lambda^{g\vec{l'}r}} = \left(\prod_{k=0}^{r-1} \delta_{gl_k-\mu_k+\nu_k,gl_k'}\right) \delta_{(n-\mu_r+\nu_r-S_{g\vec{l}}),(n-S_{g\vec{l'}})}  \left(\prod_{k=r}^{\DP-2} \delta_{gl_k-\mu_{k+1}+\nu_{k+1},gl_k'}\right),
\end{equation}
where $l_k \equiv i + r \text{ mod }\DP$ and $l_k' \equiv j + r \text{ mod }\DP$ with $i \neq j$. Since $g \geq 2t$ and $|\nu_k - \mu_k| \leq 2t$, to satisfy $gl_k + \nu_k - \mu_k = gl_k'$ we must have $g=2t$, since otherwise:
\begin{equation}
    |l_k'-l_k| \geq 1, g >2t \Rightarrow |\nu_k - \mu_k| > 2t,
\end{equation}
which is a contradiction. Consequently, for $g=2t$:
\begin{equation}
    \forall k \in [\DP]: \nu_k - \mu_k = \pm 2t \Rightarrow \forall k \in [\DP-1]:l_k' = l_k \pm 1.
\end{equation}
Without loss of generality assume that $\nu_k - \mu_k = 2t$, which implies $l_k'=l_k + 1$ and therefore $l_k' - l_k = j - i = 1$. But then for $\DP > 2$, we have that for any other $u \in [\DP]$ such that $u \neq k, \; u \neq r$ the following holds: $\nu_u - \mu_u \neq 2t$. This then implies that $j-i \neq 1$, since there is a corresponding inequality $l_u' - l_u \neq 1$ or $l_{u-1}' - l_{u-1} \neq 1$, which is a contradiction because $j-1=1$. Therefore, the condition $gl_k + \nu_k - \mu_k = gl_k'$ or $gl_k + \nu_{k+1} - \mu_{k+1} = gl_k'$ cannot be simultaneously satisfied for all $k \in [\DP-1]$, therefore we must have that $ \delta_{\lambda^{g\vec{l}r} - \mu + \nu, \lambda^{g\vec{l'}r}} = 0$, which concludes the proof of C2 for any $\DP>2$. For $\DP=2$, see \cite{aydin2024family}. \hfill $\square$

\subsubsection*{$\boxed{\mathrm{\textbf{C3}}}$}
Again, consider $\mu, \nu \in 2t$ and, without loss of generality, consider only those $\lambda$ that satisfy $\lambda - \mu \geq 0$. We have:

\begin{equation}
    \bar{\alpha}_{i,\lambda} \alpha_{i,\lambda-\mu+\nu}
    = \sum_{r,s=0}^{\DP-1} \omega_{\DP}^{i(s-r)} \sum_{\substack{\vec{l} \equiv i+r \text{ mod } \DP \\ \vec{l} \in \Rl}} \sum_{\substack{\vec{l'} \equiv i+s \text{ mod } \DP \\ \vec{l'} \in \Rl}}\bar{f}(\vec{l})f(\vec{l'})\delta_{\lambda,\lambda^{g\vec{l}r}}\delta_{\lambda^{g\vec{l}r} - \mu + \nu, \lambda^{g\vec{l'}s}}.
\end{equation}
$\bullet$ For terms with $r \neq s$:
\begin{align}
     \delta_{\lambda^{g\vec{l}r} - \mu + \nu, \lambda^{g\vec{l'}s}} = &\delta_{gl_0-\mu_0+\nu_0,gl_0'}\delta_{gl_1-\mu_1+\nu_1,gl_1'} \dots \delta_{gl_{r-1}-\mu_{r-1}+\nu_{r-1},gl_{r-1}'} \nonumber \\
     &\times \delta_{(n-\mu_r+\nu_r-S_{g\vec{l}}),gl_r'} \delta_{gl_{r}-\mu_{r+1}+\nu_{r+1},gl_{r+1}'}
     \dots \delta_{gl_{s-2}-\mu_{s-1}+\nu_{s-1},gl_{s-1}'} \\
     &\times \delta_{gl_{s-1}-\mu_s+\nu_s,(n-S_{g\vec{l'}})}\delta_{gl_s-\mu_{s+1}+\nu_{s+1},gl_s'} \dots \delta_{gl_{\DP-2}-\mu_{\DP-1}+\nu_{\DP-1},gl_{\DP-2}'}, \nonumber
\end{align}
but again due to \cref{eqn:lmodel_constraint}, we have:
\begin{eqnarray}
   gl_r'+\mu_r-\nu_r+S_{g\vec{l}} \leq gb + \mu_r - \nu_r &=& n + \mu_r - \nu_r - \delta -1 \leq n + 2t - \delta - 1 \leq n-1 < n \nonumber \\
   \Rightarrow
   \delta_{(n-\mu_r+\nu_r-S_{g\vec{l}}),gl_r'} &=& 0,
\end{eqnarray}
thus
$
    \delta_{\lambda^{g\vec{l}r}-\mu+\nu, \lambda^{g\vec{l'}s}} = 0.
$
\\
\\
$\bullet$ For terms with $r = s$:
\begin{equation}
     \delta_{\lambda^{g\vec{l}r} - \mu + \nu, \lambda^{g\vec{l'}r}} = \left(\prod_{k=0}^{r-1} \delta_{gl_k-\mu_k+\nu_k,gl_k'}\right) \delta_{(n-\mu_r+\nu_r-S_{g\vec{l}}),(n-S_{g\vec{l'}})}  \left(\prod_{k=r}^{\DP-2} \delta_{gl_k-\mu_{k+1}+\nu_{k+1},gl_k'}\right),
\end{equation}
where $l_k \equiv i + r \text{ mod }\DP$ and $l_k' \equiv i + r \text{ mod }\DP$. So the only terms that survive must satisfy $gl_k + \nu_k - \mu_k = gl_k'$ for all $k \in [\DP-1]$. Since $g \geq 2t$, $\DP \geq 2$, $l_k \equiv l_k' \mod \DP$, and $|\nu_k - \mu_k| \leq 2t$, to satisfy $gl_k + \nu_k - \mu_k = gl_k'$ we must have $\nu_k = \mu_k$ and $l_k = l_k'$ for all $k \in [\DP-1]$. Therefore only terms with $\mu = \nu$ and $\vec{l} = \vec{l'}$ survive, thus:
\begin{equation}\label{eqn:C3_simplicial_coefs}
    \bar{\alpha}_{i,\lambda} \alpha_{i,\lambda-\mu+\nu}
    = |\alpha_{i,\lambda}|^2 = \sum_{r=0}^{\DP-1} \sum_{\substack{\vec{l} \equiv i+r \text{ mod } \DP \\ \vec{l} \in \Rl}}|f(\vec{l})|^2\delta_{\lambda,\lambda^{g\vec{l}r}}.
\end{equation}
Now consider the full expression for C3 from \cref{thrm:PI_KL_conds_qudits} and substitute the expression from \cref{eqn:C3_simplicial_coefs}. We have:

\begin{multline}\label{eqn:C3_collecting_terms}
     \sum_{\l \partition{\DP} n} (\bar{\alpha}_{i, \l} \alpha_{i, \l-\mu+\nu} - \bar{\alpha}_{j, \l} \alpha_{j, \l-\mu+\nu}) \frac{\binom{n-2t}{\l-\mu}}{\sqrt{\binom{n}{\l} \binom{n}{\l-\mu+\nu}}} = \sum
     _{\l \partition{\DP} n} \frac{\binom{n-2t}{\lambda-\mu}}{\binom{n}{\lambda}} (|\alpha_{i,\lambda}|^2 - |\alpha_{j,\lambda}|^2) \\
     = \sum
     _{\l \partition{\DP} n} \frac{\binom{n-2t}{\lambda-\mu}}{\binom{n}{\lambda}} \left(\sum_{r=0}^{\DP-1} \sum_{\substack{\vec{l} \equiv i+r \text{ mod } \DP \\ \vec{l} \in \Rl}}|f(\vec{l})|^2\delta_{\lambda,\lambda^{g\vec{l}r}} - \sum_{r'=0}^{\DP-1} \sum_{\substack{\vec{l'} \equiv j+r' \text{ mod } \DP \\ \vec{l'} \in \Rl}}|f(\vec{l'})|^2\delta_{\lambda,\lambda^{g\vec{l'}r'}}\right) := C
\end{multline}
We can collect the terms in \cref{eqn:C3_collecting_terms} in such a way that $j+r' \equiv i+r \mod \DP$, i.e.\ set $r' \equiv i + r - j \mod \DP$. Then we have:
\begin{align}
    \label{eqn:C3_simplified_further}
    C &= \sum_{\l \partition{\DP} n} \frac{\binom{n-2t}{\lambda-\mu}}{\binom{n}{\lambda}} \sum_{r=0}^{\DP-1} \sum_{\substack{\vec{l} \equiv i+r \mod \DP \\ \vec{l} \in \Rl}} |f(\vec{l})|^2 (\delta_{\lambda,\lambda^{g\vec{l}r}} - \delta_{\lambda,\lambda^{g\vec{l}(i+r-j)}}) \\
    &= \sum_{r=0}^{\DP-1} \sum_{\substack{\vec{l} \equiv i+r \mod \DP \\ \vec{l} \in \Rl}} |f(\vec{l})|^2 \left( \frac{\binom{n-2t}{\lambda^{g\vec{l}r}-\mu}}{\binom{n}{\lambda^{g\vec{l}r}}} - \frac{\binom{n-2t}{\lambda^{g\vec{l}(i+r-j)}-\mu}}{\binom{n}{\lambda^{g\vec{l}(i+r-j)}}} \right).
\end{align}
By relabeling the summation index $i+r \rightarrow r$ and rearranging the terms we get:
\begin{equation}\label{eqn:C}
    C = \sum_{r=0}^{\DP-1} \sum_{\substack{\vec{l} \equiv r \text{ mod } \DP \\ \vec{l} \in \Rl}} |f(\vec{l})|^2\left( \frac{\binom{n-2t}{\lambda^{g\vec{l}(r-i)}-\mu}}{\binom{n}{\lambda^{g\vec{l}(r-i)}}} - \frac{\binom{n-2t}{\lambda^{g\vec{l}(r-j)}-\mu}}{\binom{n}{\lambda^{g\vec{l}(r-j)}}} \right).
\end{equation}
Finally, we notice that multinomial coefficients are invariant under permutations of indices in their defining composition
\begin{equation}
    \binom{n}{\lambda^{g\vec{l}(r-i)}} = \binom{n}{\lambda^{g\vec{l}(r-j)}} = \binom{n}{\lambda^{g\vec{l}r}}.
\end{equation}
Subsequently, using the relation of \cref{eqn:C4_cond_Simplicial_codes}, \cref{eqn:C} simplifies to:
\begin{equation}\label{eqn:C4_simplicial_final_form}
    C = \sum_{r=0}^{\DP-1} \sum_{\substack{\vec{l} \equiv r \text{ mod } \DP \\ \vec{l} \in \Rl}} \frac{|f(\vec{l})|^2}{\binom{n}{\lambda^{g\vec{l}r}}} \left[ \binom{n-2t}{\lambda^{g\vec{l}(r-i)} - \mu} - \binom{n-2t}{\lambda^{g\vec{l}(r-j)} - \mu} \right] = 0.
\end{equation}
\hfill $\square$
\\

\noindent Since C1, C2, and C3 are satisfied, this concludes the proof. \\
\end{proof}

Note that for a transpostion $\sigma_{r-i,r-j} \in \mathrm{S}_{\DP}$, where $\mathrm{S}_{\DP}$ is the symmetric group of size $\DP$, acting on a composition as $\sigma_{r-i,r-j} \cdot  \lambda^{g\vec{l}(r-i)} = \lambda^{g\vec{l}(r-j)}$, since permutations acting on the full multinomial coefficients leave them invariant, we also have the following equality:
\begin{equation}\label{eqn:C4_cond_permutations}
\binom{n-2t}{\lambda^{g\vec{l}(r-j)}-\mu} = \binom{n-2t}{\lambda^{g\vec{l}(r-i)}-\sigma_{r-i,r-j}\cdot\mu}.
\end{equation}

\subsection{Approximation to the optimal set $\Rl$ for simplicial PI codes}\label{ch:optimal_lmodel}
Notice that for a given $t$, \cref{eqn:C4_cond_Simplicial_codes} with $i,j \in [\DL]$ and $\mu\partition{\DP}2t$ become a linear system of equations in (non-negative) variables $|f(\vec{l})|^2$ once we specify the region $\Rl$ and $n$. Without fixing $n(t)$ and the region $\Rl$, while much simpler than the system from \cref{thrm:PI_KL_conds_qudits}, \cref{eqn:C4_cond_Simplicial_codes} is still a system of non-linear equations. If we were able to identify the optimal region $\Rl$ (in terms of giving the largest scaling of the code distance $d(n)$, i.e.\ smallest $n(t)$), then finding $f(\vec{l})$ reduces to solving a linear program \eqref{eqn:C4_cond_Simplicial_codes}. In this section we construct a good approximation to the optimal region $\Rl$ that approaches the optimal region $\Rl$ as $t$ grows.

\subsubsection{Boundary regions}
Notice how \cref{eqn:lmodel_constraint} already puts a significant constraint on what regions $\Rl$ we can choose:
\begin{equation*}
    \forall \vec{l},\vec{l'} \in \Rl, \forall k \in [\DP-1] : l_0 + l_1 + \dots +l_{k-1} + l'_k + l_{k} + \dots + l_{\DP-2} \leq b.
\end{equation*}
Next, consider $l_{\rm max}$ such that:
\begin{equation}\label{eqn:lmax_def}
    \forall \vec{l} \in \Rl, \; \forall k \in [\DP-1] : l_k \leq l_{\max}.
\end{equation}
Since $\vec{l},\vec{l'}$ in \cref{eqn:lmodel_constraint} are two independent vectors belonging to the same region, we have $l'_k \leq l_{\max}$. Because \cref{eqn:lmodel_constraint} should hold for any $\vec{l}' \in \Rl$, we can choose $l'_k = l_{\max}$. The condition in \cref{eqn:lmodel_constraint} then becomes:
\begin{equation}\label{eqn:lmodel_contraint_simplified}
    \forall \vec{l} \in \Rl :
    \begin{cases}
    l_0 + l_1 + \dots + l_{\DP-2} + l_{\max} \leq b \\
    \forall k \in [\DP-1] : l_k \leq l_{\max} 
\end{cases}
\end{equation}

From this simplified form of the constraints on the set $\Rl$ in \cref{eqn:lmodel_contraint_simplified} we can already see that the region is defined by $l_{\max}$ and $b$ only. Moreover, we can easily rescale the equations by $b$, so then $b$ would only be responsible for the volume of the region $\Rl$, while $l_{\max}$ would be effectively defining its shape. From these equations we can formulate bounds on $l_{\rm max}$ in terms of the parameter $b$ as follows.

\subsubsection*{\underline{\textit{Upper bound on $l_{\rm max}$}:}}
It is easy to see that $l_{\max}$ cannot be larger than $b/2$. To prove this, assume $l_{\max}>b/2$, then $\exists \vec{l} \in \Rl$ that has $l_{\max}$ as one of its entries. Therefore, for this $\vec{l}$ we have:
\begin{equation*}
    l_0 + l_1 + \dots + l_{\DP-2} + l_{\max} > 2l_{\max} > b,
\end{equation*}
so that \cref{eqn:lmodel_contraint_simplified} does not hold, which is a contradiction. Thus the upper bound is: $l_{\max} \leq b/2$.

For $l_{\max} = b/2$, the constraint is simply:
\begin{equation}
    l_0 + l_1 + \dots + l_{\DP-2} \leq b/2,
\end{equation}
thus the region $\Rl$ is just a simplex of dimension $\DP-1$. We show an example of $\Rl$ for $\DP=3$ with $l_{\max}=b/2$ in \cref{fig:Sierpinski_lmodel}. The codespace formed from this region (constructed via the same method as described in~\cref{sec:Intuition_Simplicial}), somewhat resembling the Sierpinski triangle, is shown in~\cref{fig:Sierpinski_2D}. We denote this set as $\Rl_S$.

\begin{figure}
  \centering
  \begin{minipage}{0.49\textwidth}
    \centering
    \begin{tikzpicture}[scale=1]
      \def\a{6}
      \draw[thick, -{Latex}] (0,0) -- ({\a+0.5},0) node[below] {$l_0$};
      \draw[thick,-{Latex}] (0,0) -- (0,{\a+0.5}) node[left]  {$l_1$};

     \coordinate (A) at (0, 0);
     \coordinate (B) at (\a, 0);
     \coordinate (C) at (0, \a); 
    \coordinate (E) at (0,\a/2);
    \coordinate (F) at (\a/2,\a/2 );
    \coordinate (G) at (\a/2, 0);

    \coordinate (O) at ({\a/2}, {\a/2});
    
    \draw[thick] (C) -- (B);
    \draw[thick] (G) -- (E);
  
     \fill[pattern={Lines[angle=45, distance=4pt, line width=0.5pt]}] (A) -- (E)-- (G) -- cycle;

    \draw[thick] (\a/2,0.1) -- (\a/2,-0.1);
    \draw[thick] (-0.1,\a/2) -- (0.1,\a/2);
    \node at (E) [left]  {$b/2$};
    \node at (G) [below]  {$b/2=l_{\rm{max}}$};
    \node at (C) [left]  {$b$};
    \node at (B) [below]  {$b$};
   \end{tikzpicture}
  \caption{Region $\Rl$ with shape defined by \mbox{$l_{\max}=b/2$} and $\DP=3$, shown as the shaded section.}
  \label{fig:Sierpinski_lmodel}
  \end{minipage}
  \hfill
  \begin{minipage}{0.49\textwidth}
    \centering
    \begin{tikzpicture}[scale=1.2]
     \def\a{6}
     \coordinate (A) at (0, 0);
     \coordinate (B) at (\a, 0);
     \coordinate (C) at (\a/2, {\a/2 * sqrt(3)}); 
     \coordinate (E) at (\a/4,{\a/2 * sqrt(3)/2});
     \coordinate (F) at ({(\a+\a/2)/2},{\a/2 * sqrt(3) / 2} );
     \coordinate (G) at (\a/2, 0);
  
     \draw[thick] (A) -- (B) -- (C) -- cycle;
     \draw[thick] (E) -- (F) -- (G) -- cycle;

     \fill[pattern=dots] (A) -- (E) -- (G) -- cycle;
     \fill[pattern=dots] (C) -- (E) -- (F) -- cycle;
     \fill[pattern=dots] (B) -- (F) -- (G) -- cycle;
    \end{tikzpicture}
    \caption{Codespace formed from $\Rl_S$ for $\DP=3$, represented as a discrete (dotted) subset of the simplex.}
    \label{fig:Sierpinski_2D}
  \end{minipage}
\end{figure}

\subsubsection*{\underline{\textit{Lower bound on $l_{\rm max}$}:}}
If we set $l_{\max} \leq  b/{\DP}$, then the first inequality in \cref{eqn:lmodel_contraint_simplified} is automatically satisfied:
\begin{equation*}
    l_{\max} \leq b/\DP \Rightarrow l_0 + l_1 + \dots + l_{\DP-2} + l_{lmax} \leq \DP \cdot l_{\max} \leq \DP \cdot b/\DP = b.
\end{equation*}
Thus for any values of $l_{\max} \leq b/\DP$ the region $\Rl$ is a hypercube with edges of length $l_{\max}$. Since the region defined by $l_{\max} = b/\DP$ contains all the other regions with smaller values of $l_{\max}$, we can set this as the lower bound for $l_{\max}$. This is because considering smaller regions will not improve the scaling of $n(t)$, since we are effectively reducing the number of variables in the relation of \cref{eqn:C4_cond_Simplicial_codes} by setting some to zero. We conclude that the lower bound is given by $l_{\max} \geq b/\DP$.

Notice that we have encountered the region $\Rl$ with $l_{\max} = b/\DP$ for the case of $\DP=3$ in \cref{fig:Mitsubishi_l_model}. The codespace formed from this region is shown in \cref{fig:Mitsubishi_2D}.
We denote this set as $\Rl_M$.

All other $\Rl$ interpolate between $\Rl_M$ and $\Rl_S$ as we change $l_{\max}$ from $b/\DP$ to $b/2$.
\\

In the next two sections we solve an \textbf{auxiliary problem} of finding the region $\Rl$ which, for a fixed $b$, has the highest volume, i.e.\ we find the optimal ratio of $l_{\max}/b$ that achieves the highest volume of $\Rl$.

We then argue in \cref{sec:justifying_optimal} why the region $\Rl$ found as a solution to this auxiliary problem is a good approximation to the optimal $\Rl$ giving the smallest scaling $n(t)$.

\subsubsection{Optimizing volume for $\DP=3$}\label{sec:Optim_vol_q=3}
Let us first consider a special case of this auxiliary problem for $\DP=3$.
\\
\\

\begin{problem}
    For $\DP=3$ find $l_{\max}$ such that for a fixed $b$ it maximizes the volume of the region $\Rl$.
\end{problem}
\noindent \textit{Solution:}\\
We have that the simplex and its region $\Rl$ have dimension $D = \DP-1=2$. The first inequality in \cref{eqn:lmodel_contraint_simplified} defines a half-plane $\mathfrak h$ cut by a line, given by the equation:
\begin{equation}\label{eqn:lmodel_line}
    \frak h: \; l_0 + l_1 \leq b-l_{\max},
\end{equation}
and the rest of the inequalities in \cref{eqn:lmodel_contraint_simplified} define a square:
\begin{equation}\label{eqn:lmodel_square}
\frak s =
    \begin{cases}
        l_0\leq l_{\max} \\
        l_1 \leq l_{\max}
    \end{cases} ,
\end{equation}
therefore any region $\Rl$ for $q=3$ is given by $ {\frak s}\,\cap \, {\frak h}$ depicted as the shaded region in \cref{fig:lmodel_intermediate}.

\begin{figure}
\centering
\begin{tikzpicture}[scale=1.5]
  \def\a{6}
  \draw[thick, -{Latex}] (0,0) -- ({\a+0.5},0) node[below] {$l_0$};
  \draw[thick,-{Latex}] (0,0) -- (0,{\a+0.5}) node[left]  {$l_1$};

  \coordinate (A) at (0, 0);
  \coordinate (B) at (\a, 0);
  \coordinate (C) at (0, \a); 
  \coordinate (E) at (0,\a/2);
  \coordinate (F) at (\a/2,\a/2 );
  \coordinate (G) at (\a/2, 0);

  \coordinate (lmaxG) at ({4/5 * \a/2}, 0);
  \coordinate (otherlmaxG) at ({6/5 * \a/2}, 0);
  
  \coordinate (lmaxE) at (0, {4/5 * \a/2});
  \coordinate (otherlmaxE) at (0, {6/5 * \a/2});
  
  \coordinate (O) at ({4/5 * \a/2}, {4/5 * \a/2});
  
  \coordinate (O1) at ({4/5 * \a/2 *1/2}, {4/5 * \a/2});
  
  \coordinate (O2) at ({4/5 * \a/2}, {4/5 * \a/2 *1/2});
  
  \draw[thick] (C) -- (B);
  \draw[thick] (lmaxG) -- (O2);
  \draw[thick] (lmaxE) -- (O1);
  \draw[thick, dashed] (otherlmaxG) -- (O2);
  \draw[thick, dashed] (otherlmaxE) -- (O1);
  \draw[thick] (O1) -- (O2);
  \draw[thick, dashed] (O1) -- (O)--(O2);

  \fill[pattern={Lines[angle=45, distance=4pt, line width=0.5pt]}] (lmaxE) -- (O1)-- (O2) -- (lmaxG) --(A) -- cycle;

  \draw[thick] (\a/2,0.1) -- (\a/2,-0.1);
  \draw[thick] (-0.1,\a/2) -- (0.1,\a/2);
  \node at (E) [left]  {$b/2$};'
  \node at (otherlmaxE) [left]  {$b-l_{\max}$};
  \node at (G) [below]  {$b/2$};
  \node at (otherlmaxG) [below]  {$\qquad b - l_{\max}$};
  \node at (C) [left]  {$b$};
  \node at (B) [below]  {$b$};
  \node at (lmaxG) [below]  {$l_{\rm{max}}$};
  \node at (lmaxE) [left]  {$l_{\rm{max}}$};
\end{tikzpicture}
\caption{Any region $\Rl$ for $q=3$ is a square at origin with edge length $l_{\max}$, cut by a line $l_1 = b-l_{\max} - l_0$.}
\label{fig:lmodel_intermediate}
\end{figure}
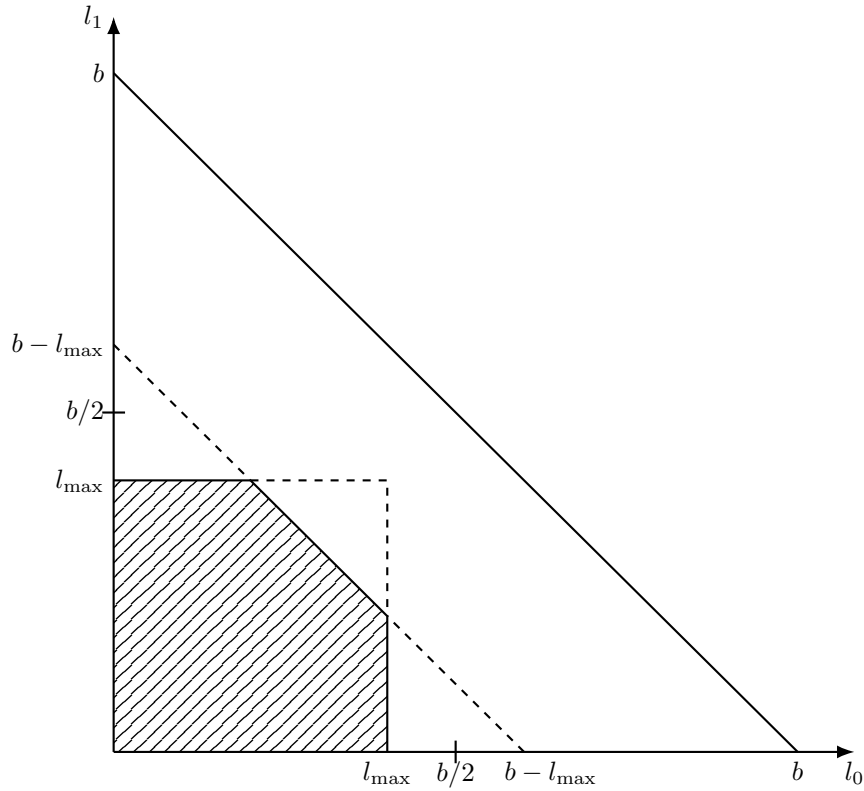
Notice that the larger $l_{\max}$ is, the larger fraction of the square is being cut by the line. Indeed, for $l_{\max}=b/q=b/3$ the line touches the square, and for $l_{\max}=b/2$ the line cuts the square along the diagonal,cf. \cref{fig:lmodel_boundaries_q=3}.

\begin{figure}
  \centering
  \begin{minipage}{0.4\textwidth}
    \centering
    \begin{tikzpicture}[scale=1]
      \def\a{6}
      \draw[thick, -{Latex}] (0,0) -- ({\a+0.5},0) node[below] {$l_0$};
      \draw[thick,-{Latex}] (0,0) -- (0,{\a+0.5}) node[left]  {$l_1$};

      \coordinate (A) at (0, 0);
      \coordinate (B) at (\a, 0);
      \coordinate (C) at (0, \a); 
      \coordinate (E1) at (0,\a/3);
      \coordinate (E2) at (0,2*\a/3);
  
      \coordinate (F1) at (\a/3,2*\a/3 );
      \coordinate (F2) at (2*\a/3,\a/3 );
      
      \coordinate (G1) at (\a/3, 0);
      \coordinate (G2) at (2*\a/3, 0);

      \coordinate (O) at ({\a/3}, {\a/3});
        
      \draw[thick] (C) -- (B);
      \draw[thick] (G1) -- (O);
      
      \draw[thick] (E1) -- (O);
      \draw[thick, dashed] (G2) -- (E2);

      \fill[pattern={Lines[angle=45, distance=4pt, line width=0.5pt]}] (A) -- (E1)-- (O) --(G1)-- cycle;

      \draw[thick] (\a/3,0.1) -- (\a/3,-0.1);
      \draw[thick] (-0.1,\a/3) -- (0.1,\a/3);
      \node at (E1) [left]  {$b/3$};
      \node at (G1) [below]  {$b/3=l_{\rm{max}}$};
      \node at (C) [left]  {$b$};
      \node at (B) [below]  {$b$};
    \end{tikzpicture}
    \caption*{Graphical representation of $\Rl_M$.}
  \end{minipage}
  \hfill
  \begin{minipage}{0.4\textwidth}
    \centering
    \begin{tikzpicture}[scale=1]
      \def\a{6}
      \draw[thick, -{Latex}] (0,0) -- ({\a+0.5},0) node[below]{$l_0$};
      \draw[thick,-{Latex}] (0,0) -- (0,{\a+0.5}) node[left]{$l_1$};
  
      \coordinate (A) at (0, 0);
      \coordinate (B) at (\a, 0);
      \coordinate (C) at (0, \a); 
      \coordinate (E) at (0,\a/2);
      \coordinate (F) at (\a/2,\a/2 );
      \coordinate (G) at (\a/2, 0);

      \coordinate (O) at ({\a/2}, {\a/2});
        
      \draw[thick] (C) -- (B);
      \draw[thick, dashed] (G) -- (E);
      
      \draw[thick, dashed] (E) -- (O);
      \draw[thick, dashed] (G) -- (O);
      
      \fill[pattern={Lines[angle=45, distance=4pt, line width=0.5pt]}] (A) -- (E)-- (G) -- cycle;
  
      \draw[thick] (\a/2,0.1) -- (\a/2,-0.1);
      \draw[thick] (-0.1,\a/2) -- (0.1,\a/2);
      \node at (E) [left]  {$b/2$};
      \node at (G) [below]  {$b/2=l_{\rm{max}}$};
      \node at (C) [left]  {$b$};
      \node at (B) [below]  {$b$};
    \end{tikzpicture}
    \caption*{Graphical representation of $\Rl_S$.}
  \end{minipage}
  \caption{Boundary regions presented as squares cut by lines.}
  \label{fig:lmodel_boundaries_q=3}
\end{figure}
For $l_{\max} \in [b/3;b/2]$, the volume of the region $\Rl$ (i.e.\ the area of the shaded region in \cref{fig:lmodel_intermediate}) reads
\begin{equation}
V=l_{\max}^2 - (l_{\max} - (b - 2l_{\max}))^2 \cdot \frac{1}{2}=l_{\max}^2 - (3l_{\max}-b)^2 \cdot \frac{1}{2}.
\end{equation}
It follows that $V$ is maximized for
\begin{equation} \label{eqn:lmax_opt}
    l_{\max} = \frac{3}{7}b.
\end{equation}
\qed

\subsubsection{Optimizing the volume for arbitrary $\DP$}\label{sec:Optim_vol_all_q}

~\\

\noindent 
\begin{problem} \label{prob:max_vol_R}
     For a fixed $b$, find $l_{\rm max}$ that maximizes the volume of the region $\Rl$. \label{prob:max_vol_R}
\end{problem}
\noindent \textit{Solution:}\\
For arbitrary $\DP$, it is easy to see that the $D:=\DP-1$-dimensional system \eqref{eqn:lmodel_contraint_simplified} defines a hypercube cut by a hyperplane:
\begin{equation}\label{eqn:hypercube_hyperplane_general}
    \forall \vec{l} \in \Rl :
    \begin{cases}
    l_0 + l_1 + \dots + l_{\DP-2} \leq b - l_{\max} \\
    \forall k \in [\DP-1] : l_k \leq l_{\max} 
\end{cases},
\end{equation}
so first we have to compute the volume of this region. It turns out that this is not as simple as in the $\DP=3$ case: in general, computing the volume of convex polytopes is as hard as computing the permanent of a matrix, i.e.\ this problem is \textbf{\#P}-hard \cite{dyer1988polyhedra}. Even in the case of a hypercube cut by a single hyperplane, this problem is algorithmically difficult \cite{khachiyan1989problem}. However, there exist analytic expressions for this problem, which in certain special cases simplify drastically. Since our equation for the hyperplane (first line of \cref{eqn:hypercube_hyperplane_general}) is quite simple, fortunately the volume formula simplifies in our case as well.

We can compute the volume of a hypercube cut by a hyperplane using the following formula \cite{cho2022volumehypercubesclippedhyperplanes, Barrow01011979, polya1913berechnung}:
\begin{equation}\label{eqn:volume_formula}
    \mathrm{vol}([0,1]^D \cap H^+) = \sum_{\mathbf{v} \in F^0 \cap H^+} \frac{(-1)^{|0_{\mathbf{v}}|} h(\mathbf{v})^D}{D! \prod_{t=1}^D a_t},
\end{equation}
where:
\begin{itemize}[itemsep=1pt, parsep=1pt, topsep=1pt, partopsep=1pt]
    \item $D=\DP-1$;
    \item The hypercube $[0,1]^D$ has edges of length 1;
    \item $F^0$ is the set of all vertices of the hypercube;
    \item $|0_\mathbf{v}|$ is the number of zeros in the entries of $\mathbf{v}$;
    \item The half-space $H^+$ clipped by the hyperplane is generally given by:
    \begin{equation}
        H^+ := \left\{ \mathbf{l} \,|\, h(\mathbf{l}) := \mathbf{a} \cdot \mathbf{l} + r = a_0 l_0 + a_1 l_1 + \cdots + a_{D-1} l_{D-1} + r \geq 0 \right\},
    \end{equation}
    where, in our case, $r=\frac{b}{l_{\max}}-1$ and $\mathbf{a} = (\underbrace{-1,\dots,-1}_D)$.
\end{itemize}
The reason why $r=\frac{b}{l_{\max}}-1$ instead of $b-l_{\max}$ like in the system \eqref{eqn:lmodel_contraint_simplified} is because we have rescaled the hypercube by $l_{\max}$ to $[0,1]^D$ in order to fit the volume formula, so we had to rescale the hyperplane as well. The volume formula simplifies further after the substitutions:
\begin{equation*}
    \mathrm{vol}([0,1]^D \cap H^+) = \sum_{\mathbf{v} \in F^0 \cap H^+} \frac{(-1)^{|0_{\mathbf{v}}|} h(\mathbf{v})^D}{D! \, (-1)^D} = \sum_{\mathbf{v} \in F^0 \cap H^+}(-1)^{|1_{\mathbf{v}}|} \frac{h(\mathbf{v})^D}{D!},
\end{equation*}
where $|1_{\mathbf{v}}|$ is the number of ones in the entries of $\mathbf{v}$. Note that $h(\textbf{v})$ also simplifies for vertices:
\begin{equation}
    h(\textbf{v}) = \frac{b}{l_{\max}}-1-|1_{\textbf{v}}|
\end{equation}
Now, to find the volume of the hypercube of edge size $l_{\max}$ cut by a hyperplane, we rescale the system back by $l_{\max}$, to get:
\begin{equation}\label{eqn:lmodel_volume_all_q}
    V = \sum_{\mathbf{v} \in F^0 \cap H^+}(-1)^{|1_{\mathbf{v}}|} \frac{h(\mathbf{v})^D}{D!}l_{\max}^D = \sum_{\mathbf{v} \in F^0 \cap H^+}(-1)^{|1_{\mathbf{v}}|} \frac{(b-(1+|1_{\textbf{v}}|)l_{\max})^D}{D!}
\end{equation}
Notice that vertices $\textbf{v}$ of the hypercube that also belong to the half-space $H^+$ must satisfy $h(\textbf{v}) = \frac{b}{l_{\max}}-1-|1_{\textbf{v}}| \geq 0$, i.e.\ $|1_{\textbf{v}}| \leq  \frac{b}{l_{\max}}-1$. Additionally, there are multiple terms in the above expression that have different $\textbf{v}$ but same $|1_{\textbf{v}}|$, specifically there are exactly $\binom{D}{|1_{\textbf{v}}|}$ such terms for every $|1_{\textbf{v}}|$. So we can further simplify \cref{eqn:lmodel_volume_all_q} to:
\begin{equation}\label{eqn:volume_func}
    V = \sum_{w \in \{0,1,\dots,\left \lfloor{r}\right \rfloor \}}(-1)^w \frac{(b-(1+w)l_{\max})^D}{D!} \binom{D}{w}
\end{equation}
where, again, $r = \frac{b}{l_{\max}}-1$ and $l_{\max} \in [b/\DP,b/2]$. Therefore in this special case we have a closed form formula for the volume of a convex polytope. Now, to find $l_{\max}$ that maximizes the volume, we again solve to find all extrema:
\begin{equation}\label{eqn:dV_dl_simplified}
    \frac{dV}{dl_{\max}} =\sum_{w \in \{0,1,\dots,\left \lfloor{r}\right \rfloor \}}(-1)^{1+w} (1+w)\frac{(b-(1+w)l_{\max})^{D-1}}{(D-1)!} \binom{D}{w} = 0.
\end{equation}
The first couple of solutions found numerically are shown in \cref{tab:lmodel_optim_solutions}. $V(l_{\max})$ is plotted in range $l_{\max} \in [b/\DP,b/2]$ for several values of $\DP$ in \cref{fig:vol_q_3-9}
\begin{table}[H]
    \centering
    \begin{tabular}{ |c|c|c| } 
      \hline
       &    &        \\ 
       $D=2\; (\DP=3) \!: \;\frac{l_{\max}}{b} = \frac{3}{7}$ 
       &  $D=3\; (\DP=4) \!: \;\frac{l_{\max}}{b} = \frac{1}{23}(11-\sqrt{6})$ 
            & $D=4\; (\DP=5) \!: \;\frac{l_{\max}}{b} =
        \frac{1}{3}$ \\ 
       &    &        \\ 
      \hline
    \end{tabular}
    \caption{Optimal $l_{\max}$ for $\DP=3,4,5$.}
    \label{tab:lmodel_optim_solutions}
\end{table}

\begin{figure}[H]
    \centering
    \includegraphics[width=0.7\linewidth]{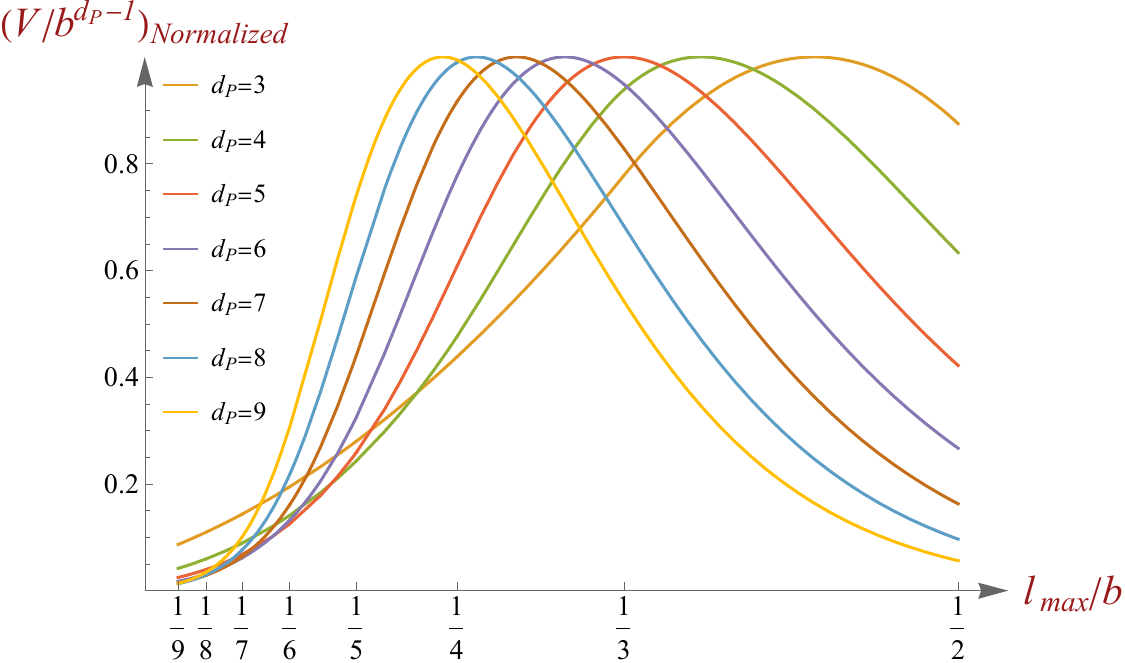}
    \caption{Plots of $V(l_{\max}/b)$ for various values of $\DP$. Note that we rescaled the volume by $b^{\DP-1}$, and also for each curve we normalized by the value of $V$ at the extremum (found through solving \cref{eqn:dV_dl_simplified}) for each value of $\DP$ respectively, in order to be able to plot all the curves on the same scale. We had to do this since as $\DP$ grows for fixed $b$ the value of $V(l_{\max}/b)$ decays super-exponentially, in other words the ratio of the volume of $\Rl$ to the volume of the simplex decays super-exponentially with the dimension of the simplex, which happens due to the curse of dimensionality \cite{bellman1961adaptive}. Notice that for any value of $\DP$, in the range $l_{\max} \in [b/\DP,b/2]$ it is apparent that $V(l_{\max})$ always has one maximum, which is the solution we find when solving \cref{eqn:dV_dl_simplified}.}
    \label{fig:vol_q_3-9}
\end{figure}
\qed

\subsection{Justifying the approximation}\label{sec:justifying_optimal}

In Appendices~\ref{sec:Optim_vol_q=3} and \ref{sec:Optim_vol_all_q} we described the optimal region in terms of giving the highest volume for fixed $b$. But notice that in the system of equations \eqref{eqn:C4_cond_Simplicial_codes} for fixed $n$ (so for fixed $b$ as well) the more possible different vectors $\vec{l}$ the region contains, the more terms does each equation have. That is, the number of variables grows with the number of possible vectors $\vec{l}$ in the region. And the number of vectors grows with the volume of the region $\Rl$.

Of course, if there was a guarantee that all the equations in the system are linearly independent, solving \cref{prob:max_vol_R} would obviously give the optimal $\Rl$ for the original problem: since the number of equations in the system \eqref{eqn:C4_cond_Simplicial_codes} is a function of $\DP,\DL,t$ only, and $n$ is a function of $t$ and $\Rl$, there is clearly some minimal amount of variables that we have to keep non-zero for the linear program (LP) to be solvable. That is, for fixed $t$ and fixed $\Rl$, there is a lower bound on the amount of $f(\vec{l})$ that have to be non-zero, therefore a lower bound on the volume of the region, which for a fixed $l_{\max}$ scales only with $n$. This means that we can achieve lower $n$ (that still keeps the LP solvable) by choosing the right $l_{\max}$ to maximize the volume. 

Unfortunately, it is non-trivial whether the equations in the system \eqref{eqn:C4_cond_Simplicial_codes} are linearly independent. Nevertheless, we can use numerical methods to test the idea of optimizing for the number of variables to find the optimal solution in terms of scaling. The numerical evidence we produced (see \cref{fig:n_lmax_convex} and \cref{fig:n_lmax_converge}) suggests that this approach gives a good approximation.

From \cref{fig:n_lmax_convex} we can make several observations. First, it suggests that the dependence of $n_{\rm min}(l_{\rm max})$ is convex on $l_{max}/b \in [1/\DP, 1/2]$. This is something we should expect if we assume that solutions found through \cref{prob:max_vol_R} are good approximations, since the volume function \cref{eqn:volume_func} is concave on that range of $l_{max}/b$. Second, we see that that the value of $l_{\rm max}$ given by the continuous \cref{prob:max_vol_R} is close to the true solution for various values of $\DP$ and $t$. Note that all values of $l_{max}$ and $b$ are integers for the true problem of testing $n(t)$ dependence of simplicial codes against $l_{max}$, while in the auxiliary \cref{prob:max_vol_R} $l_{max}$ is continuous.

Furthermore, \cref{fig:n_lmax_converge} suggests that for a fixed $\DP$, as we increase $t$, the solution for $l_{max}/b$ found through solving \cref{prob:max_vol_R} for large $t$ is within range of possible true solutions for $l_{max}/b$ achieving the minimum $n(t)$, while all the points outside of the solution for \cref{prob:max_vol_R} that also achieve the minimum for smaller values of $t$ tend to diverge from the minimum as $t$ increases. In other words, as the range of possible $l_{max}$ that achieve the minimum $n(t)$ decreases as $t$ increases, the solution for \cref{prob:max_vol_R} stays within that range. Note that this is not necessarily true for smaller values of $t$, as observed in \cref{fig:n_lmax_convex} for $\DP = 4$ and $5$.

\addtocounter{figure}{-1}
\begin{figure}[H]
\centering
\begin{tabular}{cc}
\begin{subfigure}{0.47\textwidth}
    \includegraphics[width=\linewidth]{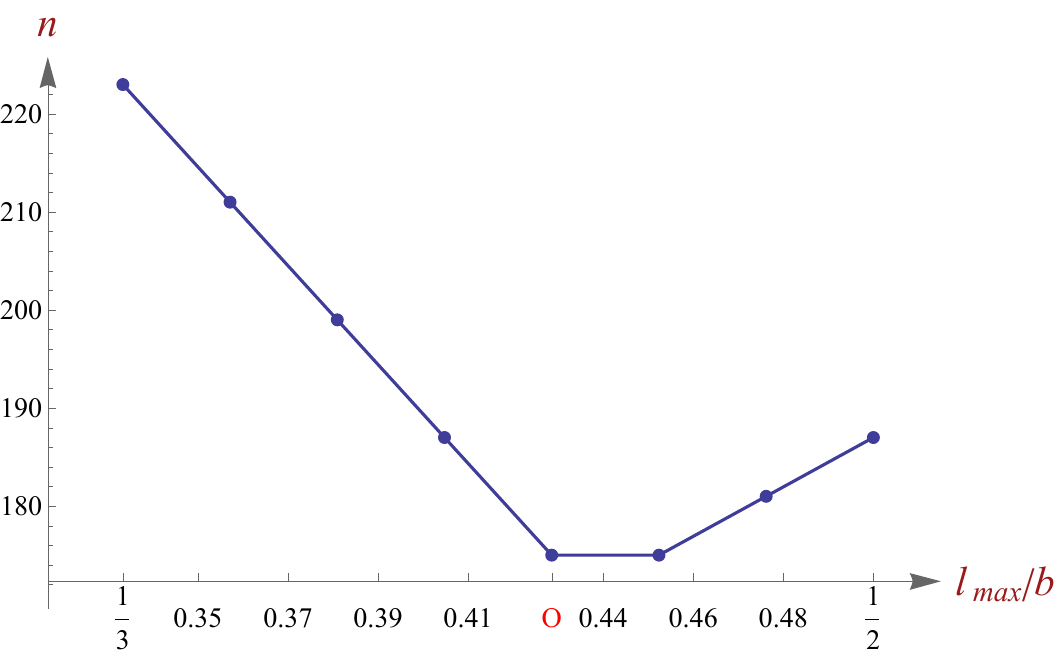}
    \caption{$n(l_{\max})$ at $\DP=3,t=3$, step = 1/42}
\end{subfigure}
&
\begin{subfigure}{0.47\textwidth}
    \includegraphics[width=\linewidth]{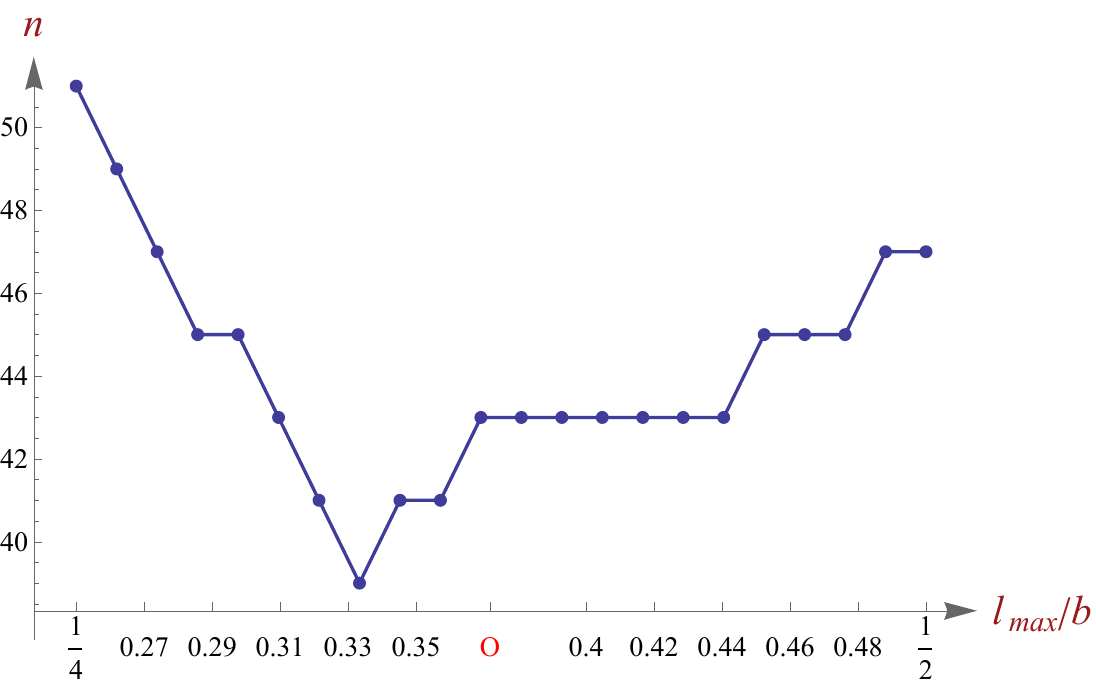}
    \caption{$n(l_{\max})$ at $\DP=4,t=1$, step = 1/84}
\end{subfigure}
\\[0.4cm]
\begin{subfigure}{0.47\textwidth}
    \includegraphics[width=\linewidth]{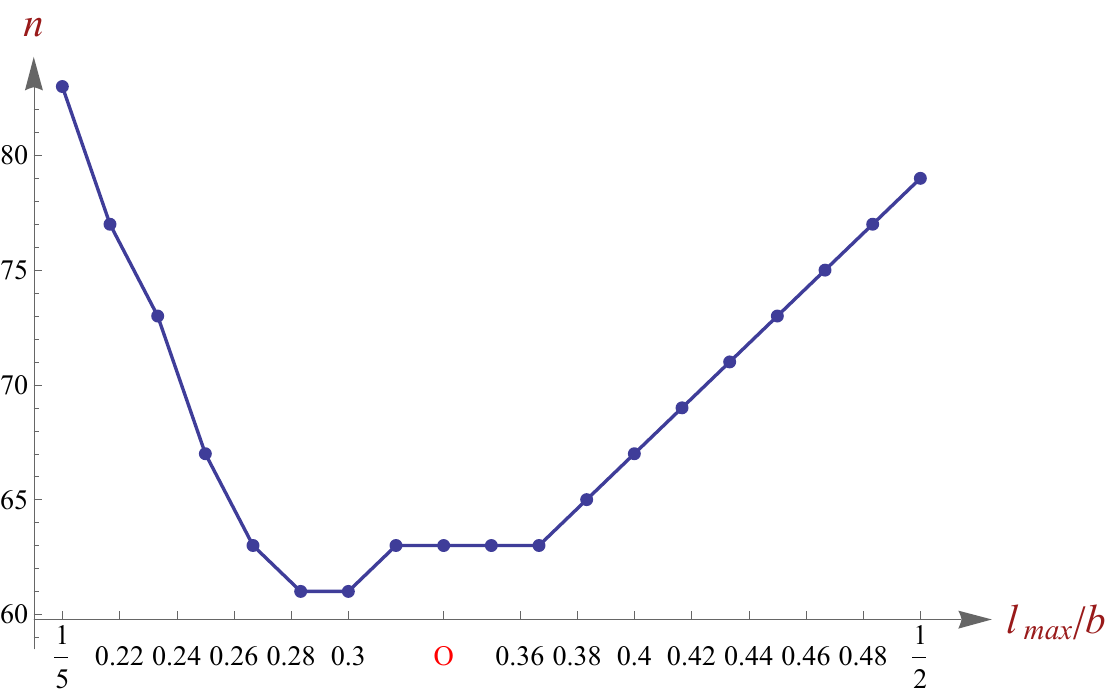}
    \caption{$n(l_{\max})$ at $\DP=5,t=1$, step = 1/60}
\end{subfigure}
&
\parbox[b]{0.6\textwidth}{
\captionof{figure}{Dependence of the block length $n$ on the shape parameter $l_{\max}$. Every individual $n(l_{\max})$ point on the plot was obtained by numerically solving the linear program obtained by fixing $n,t,\DP,l_{\max}$ and decreasing $n$ until there is no solution to the LP, and labelling the smallest $n$ with a solution as $n(l_{\max})$ for given $l_{\max}$. Here, ``step'' is the step in the plotting range of $l_{\max}$, and \textcolor{red}{O} denotes the predicted extremal $l_{\max}$ value, obtained via solving \cref{prob:max_vol_R}.}
\label{fig:n_lmax_convex}
}
\end{tabular}
\end{figure}

\addtocounter{figure}{-1}
\begin{figure}[H]
\centering
\begin{tabular}{cc}
\begin{subfigure}{0.47\textwidth}
    \includegraphics[width=\linewidth]{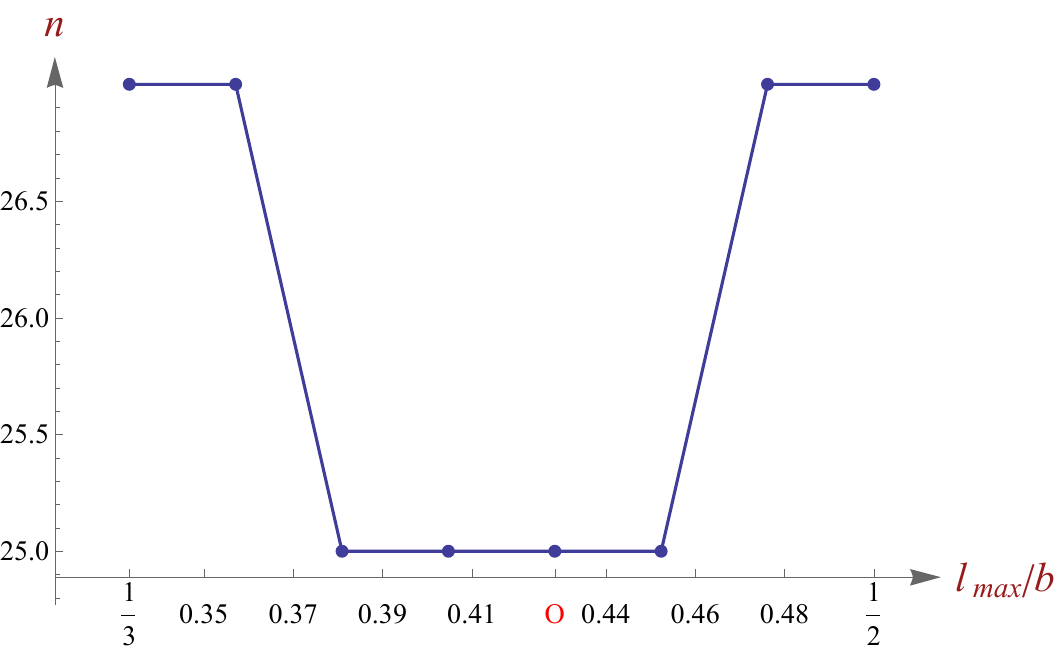}
    \caption{$n(l_{\max})$ at $\DP=3,t=1$}
\end{subfigure}
&
\begin{subfigure}{0.47\textwidth}
    \includegraphics[width=\linewidth]{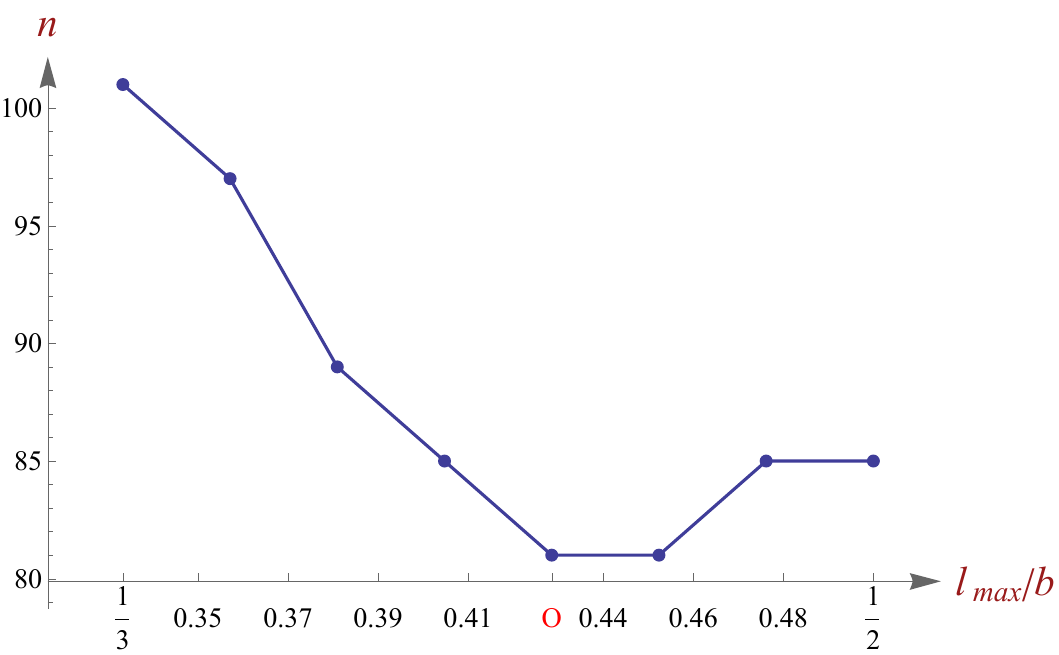}
    \caption{$n(l_{\max})$ at $\DP=3,t=2$}
\end{subfigure}
\\[0.4cm]
\begin{subfigure}{0.47\textwidth}
    \includegraphics[width=\linewidth]{ManuscriptFigures/n_lmax_q3_t3.pdf}
    \caption{$n(l_{\max})$ at $\DP=3,t=3$}
\end{subfigure}
&
\parbox[b]{0.6\textwidth}{
\captionof{figure}{Dependence of the block length $n$ on the shape parameter $l_{\max}$. Every individual $n(l_{\max})$ point on the plot was obtained by solving an LP, in exactly the same way as described for \cref{fig:n_lmax_convex}.}
\label{fig:n_lmax_converge}
}
\end{tabular}
\end{figure}

There is another more fundamental reason why the highest volume regions $\Rl$ for fixed $b$ can only be an approximation to true optimal regions. This has to do with the fact that in Appendix subsections~\ref{sec:Optim_vol_q=3} and \ref{sec:Optim_vol_all_q} we were solving a continuous optimization problem, while the initial problem of finding an optimal region for achieving minimal scaling of $n(t)$ is discrete. In this sense, largest volume for fixed $b$ achieved by choosing a certain $l_{\max}$ in the continuous optimization problem \cref{prob:max_vol_R} might not give as many terms as some other $l_{\max}$ in the discrete case. Nevertheless, the discrete case approaches the continuous limit as $t$ grows, since the number of equations grows with $t$, thus requiring more variables to keep the LP solvable. Indeed, this is something that we observe numerically in \cref{fig:n_lmax_converge}.

\subsection{Minimal scaling $n(t)$ for simplicial codes with $\DP=3$.} \label{sec:numerical_benchmark}

The way we determine the block length scaling $n(t)$ of the PI codes of \cref{eqn:Simplicial_codes} is analogous to the method we used in \cref{sec:qubit_PI_numeric_properties}. Similarly to the qubit case, we assume a polynomial dependence $n(t) = {\rm poly}(t)$. Iteratively, for every value of $t$ we try to solve \cref{eqn:C4_cond_Simplicial_codes} for increasing values of $n$ until a solution exists. For a given $t$, we choose as a starting value of $n$ to be the (minimal) solution of the problem with $t-1$ (the solution in the sense of \eqref{eq:f_cost}).
We recall from \cref{ch:optimal_lmodel} that for a given $t$ and the region $\Rl$ defined by the parameter $l_{\rm max}$, cf.\ \cref{eqn:lmodel_contraint_simplified}, $n$ is defined as:
\begin{equation}
    n = 2bt + 2t + 1.
\end{equation}
We fix $l_{\rm max}/b$ to $l_{\rm max}/b = 3/7$, which is the optimal value for $\DP=3$, cf.\ \cref{eqn:lmax_opt}. The results are shown in \cref{tab:n_b_t_Simplicial_codes}. From our numerical solution we obtain $b(t) = t^2/2 + 13t/{2} + 4$ which implies
\begin{equation}\label{eqn:scaling_n_t}
    n(t) = t^3 + 13t + 10t + 1.
\end{equation}
\begin{table}[H]
    \centering
    \begin{tabular}{|c|c|c|c|c|c|}
        \hline
         t = &  1&  2&  3&  4&  5 \\
         \hline
         b = & 11& 19& 28& 38& 49 \\
         \hline
         n = & 25& 81&175&313&501 \\
         \hline
    \end{tabular}
    \caption{First five numerical values found for the smallest $b$ for a given $t$ through solving \cref{eqn:C4_cond_Simplicial_codes}. The values of $n$ are reconstructed from $b$ and $t$.}
    \label{tab:n_b_t_Simplicial_codes}
\end{table}
We see that the scaling of the $\DP=3$ simplicial PI codes \cref{eqn:Simplicial_codes} is worse than that of Ouyang \cite{ouyang2017permutation}, for which $n = (2t+1)^2 (\DL - 1) \propto t^2$. We refer to \cref{ch:Comparing_simplicial_to_Ruskai} for further details on why we think our construction turned out to be suboptimal.

\end{document}